\documentclass[11pt,a4paper]{article}

\usepackage{preamble-fv-simple}
\usepackage{macros}

\author[1]{Lars Jaffke} 
\author[1]{Paloma T.\ Lima}
\author[2]{Geevarghese Philip}

\affil[1]{University of Bergen, Bergen, Norway} 
\affil[ ]{\texttt{\{lars.jaffke,paloma.lima\}@uib.no}} 

\affil[2]{Chennai Mathematical Institute, Chennai, India and UMI ReLaX}
\affil[ ]{\texttt{gphilip@cmi.ac.in}}

\title{Structural Parameterizations of Clique Coloring\thanks{L.\ J.\ 
	is supported by the Trond Mohn Foundation (TMS).
	The work was partially done while L.\ J.\ and P.\ T.\ L.\ were visiting Chennai Mathematical Institute.}} 

\begin{document}

\maketitle

\begin{abstract}	
	A clique coloring of a graph is an assignment of colors to its vertices
	such that no maximal clique is monochromatic.
	We initiate the study of structural parameterizations of the \cliquecol problem
	which asks whether a given graph has a clique coloring with $q$ colors.
	For fixed $q \ge 2$, we give an $\Ostar(q^{\treewidth})$-time algorithm when the input 
	graph is given together with one of its tree decompositions of width $\treewidth$.
	We complement this result with a matching lower bound under the Strong Exponential Time Hypothesis.
	We furthermore show that (when the number of colors is unbounded)
	\cliquecol is \XP parameterized by clique-width.
\end{abstract}

\section{Introduction}
Vertex coloring problems are central in algorithmic graph theory, and appear in many variants.
One of these is \cliquecol, which given a graph $G$ and an integer $k$ asks whether $G$ has a clique coloring with $k$ colors,
i.e.\ whether each vertex of $G$ 
can be assigned one of $k$ colors such that there is no monochromatic maximal clique.
The notion of a clique coloring of a graph was introduced in 1991 by Duffus et al.~\cite{DuffusEtAl1991},
and it behaves quite differently from the classical notion of a proper coloring, 
which forbids monochromatic edges.
Any proper coloring is a clique coloring, but not vice versa.
For instance, a complete graph on $n$ vertices only has a proper coloring with $n$ colors,
while it has a clique coloring with two colors.
Moreover, proper colorings are closed under taking subgraphs.
On the other hand, removing vertices or edges from a graph may introduce new maximal cliques, 
therefore a clique coloring of a graph 
is not always a clique coloring of its subgraphs, not even of its induced subgraphs.

Also from a complexity-theoretic perspective, \cliquecol behaves very differently from \graphcol.
Most notably, while it is easy to decide whether a graph has a proper coloring with two colors,
Bacs\'{o} et al.~\cite{BacsoEtAl2004} showed that it is already \coNP-hard to decide 
if a given coloring with two colors is a clique coloring.
Marx~\cite{Marx2011} later proved \cliquecol to be $\Sigma_2^p$-complete for every fixed number of (at least two) colors.

On the positive side, Cochefert and Kratsch showed that the \cliquecol problem can be solved in $\Ostar(2^n)$ 
time,\footnote{The $\Ostar$-notation suppresses polynomial factors in the input size, 
i.e.\ for inputs of size $n$, we have that $\Ostar(f(n)) = \calO(f(n)\cdot n^{\calO(1)})$.}
and the problem has been shown to be polynomial-time solvable on several graph classes.
Mohar and Skrekovski~\cite{MS1999} showed that all planar graphs are $3$-clique colorable,
and Kratochv\'{i}l and Tuza gave an algorithm that decides whether a 
given planar graph is $2$-clique colorable~\cite{KT2002}.
For several graph classes it has been shown that all their members except odd cycles on at least five vertices 
(which require three colors)
are $2$-clique colorable~\cite{BacsoEtAl2004,BacsoTuza2009,CamposEtAl2008,CerioliKorenchendler2009,Defossez2006,KleinMorgana2012,Penev2016,ShanEtAl2014}.
Therefore, on these classes \cliquecol is polynomial-time solvable.
Duffus et al.~\cite{DuffusEtAl1991} even conjectured in 1991 that perfect graphs are $3$-clique colorable, 
which was supported by many subclasses of perfect graphs being shown to be 
$2$- or $3$-clique colorable~\cite{AndreaeEtAl1991,BacsoEtAl2004,ChudnovskyLo2017,Defossez2006,DuffusEtAl1991,MS1999,Penev2016}.
However, in 2016, Charbit et al.~\cite{CharbitEtAl2016} showed that there are 
perfect graphs whose clique colorings require an unbounded number of colors.

In this work, we consider \cliquecol from the viewpoint of 
parameterized algorithms and complexity~\cite{CyganEtAl2015,DowneyFellows2013}.
In particular, we consider structural parameterizations of \cliquecol by two of the most commonly used
decomposition-based width measures of graphs, 
namely \emph{treewidth} and \emph{clique-width}.
Informally speaking, the treewidth of a graph $G$ measures how close $G$ is to being a forest.
On dense graphs, the treewidth is unbounded, and
clique-width can be viewed as an extension of treewidth that remains bounded on several
simply structured dense graphs.

Our first main result is a fixed-parameter tractable algorithm for \qcliquecol
parameterized by treewidth.  More precisely: we show that for any fixed
$q \ge 2$, \qcliquecol (asking for a clique coloring with $q$ colors) can be
solved in time $\Ostar(q^{\treewidth})$, where $\treewidth$ denotes the width of
a given tree decomposition of the input graph. We also show that this running
time is likely the best possible in this parameterization; we prove that under
the Strong Exponential Time Hypothesis (\SETH), for any~$q \ge 2$, there is
no~$\epsilon > 0$ such that \qcliquecol can be solved in time
$\Ostar((q-\epsilon)^\treewidth)$.  In fact, we rule out
$\Ostar((q-\epsilon)^t)$-time algorithms for a much smaller class of graphs than
those of treewidth $t$, namely: graphs that have both \emph{pathwidth} and
\emph{feedback vertex set number} simultaneously bounded by $t$.

Our second main result is an $\XP$ algorithm for \cliquecol with clique-width
as the parameter. The algorithm runs in time~$n^{f(w)}$, where \(w\) is the
clique-width $w$ of a given clique decomposition of the input $n$-vertex graph
and $f(w) = 2^{2^{\calO(w)}}$.  
The double-exponential dependence on $w$ in the degree of the polynomial stems
from the notorious property of clique colorings which we mentioned above;
namely, that taking induced subgraphs does not necessarily preserve clique colorings.  
This results in a large amount of information that needs to be
carried along as the algorithm progresses.

The paper is organized as follows.
In Section~\ref{sec:preliminaries}, we give introduce the basic concepts
that are important in this work;
in Section~\ref{sec:treewidth} we give the results on \qcliquecol 
parameterized by treewidth, 
and in Section~\ref{sec:cliquewidth} we give the algorithm for
\cliquecol parameterized by clique-width.

\section{Preliminaries}\label{sec:preliminaries}
\paragraph*{Graphs.}
All graphs considered here are simple and finite. 
For a graph $G$ we denote by $V(G)$ and $E(G)$ the vertex set and edge set of $G$, respectively.
For an edge $e = uv \in E(G)$, we call $u$ and $v$ the \emph{endpoints} of $e$ and we write $u \in e$ and $v \in e$. 

For two graphs $G$ and $H$, we say that $G$ is a \emph{subgraph} of $H$, written $G \subseteq H$,
if $V(G) \subseteq V(H)$ and $E(G) \subseteq E(H)$.
For a set of vertices $S \subseteq V(G)$, the \emph{subgraph of $G$ induced by $S$}
is $G[S] \defeq (S, \{uv \in E(G) \mid u, v \in S\})$.

For a graph $H$, we say that a graph $G$ is \emph{$H$-free} 
if $G$ does not contain $H$ as an induced subgraph.
For a set of graphs $\calH$, we say that $G$ is \emph{$\calH$-free}
if $G$ is $H$-free for all $H \in \calH$.

For a graph $G$ and a vertex $v \in V(G)$, the set of its \emph{neighbors} is 
$N_G(v) \defeq \{u \in V(G) \mid uv \in E(G)\}$.
Two vertices $u, v \in V(G)$ are called \emph{false twins} if $N_G(u) = N_G(v)$. 
We say that a vertex $v$ is \emph{complete} to a set $X \subseteq V(G)$ if $X \subseteq N_G(v)$.
The \emph{degree} of $v$ is $\deg_G(v) \defeq \card{N_G(v)}$.
The \emph{closed neighborhood} of $v$ is $N_G[v] \defeq \{v\} \cup N_G(v)$.
For a set $X \subseteq V(G)$, we let $N_G(X) \defeq \bigcup_{v \in X} N_G(v) \setminus X$
and $N_G[X] \defeq X \cup N_G(X)$.
In all these cases, we may drop $G$ as a subscript if it is clear from the context.
A graph is called \emph{subcubic} if all its vertices have degree at most three.

A graph $G$ is \emph{connected} if for all $2$-partitions $(X, Y)$ of $V(G)$ 
with $X \neq \emptyset$ and $Y \neq \emptyset$,
there is a pair $x \in X$, $y \in Y$ such that $xy \in E(G)$.
A \emph{connected component} of a graph is a maximal connected subgraph.
A connected graph is called a \emph{cycle} if all its vertices have degree two.
A graph that does not contain a cycle as a subgraph is called a \emph{forest}
and a connected forest is a \emph{tree}.
In a tree $T$, the vertices of degree one are called the \emph{leaves} of $T$,
denoted by $\leaves(T)$, and the vertices in $V(T) \setminus \leaves(T)$ are 
the \emph{internal vertices} of $T$.
A tree of maximum degree two is a \emph{path} and the leaves of 
a path are called its \emph{endpoints}.
A tree $T$ is called a \emph{caterpillar} if it contains a path $P \subseteq T$ such that
all vertices in $V(T) \setminus V(P)$ are adjacent to a vertex in $P$.
A forest is called a \emph{linear forest} if all its components are paths
and a \emph{caterpillar forest} if all its components are caterpillars.

A tree $T$ is called \emph{rooted}, if there is a distinguished vertex $r \in V(T)$,
called the \emph{root} of $T$, inducing an ancestral relation on $V(T)$:
for a vertex $v \in V(T)$, if $v \neq r$, the neighbor of $v$ on the path 
from $v$ to $r$ is called the \emph{parent} of $v$, and all other neighbors of $v$ 
are called its \emph{children}.
For a vertex $v \in V(T) \setminus \{r\}$ with parent $p$, 
the \emph{subtree rooted at $v$}, denoted by $T_v$, 
is the subgraph of $T$ induced by all vertices 
that are in the same connected component of 
$(V(T), E(T) \setminus \{vp\})$ as $v$. 
We define $T_r \defeq T$.

A set of vertices $S \subseteq V(G)$ of a graph $G$ is called an \emph{independent set}
if $E(G[S]) = \emptyset$.
A set of vertices $S \subseteq V(G)$ is a \emph{vertex cover} in $G$ if $V(G) \setminus S$ 
is an independent set in $G$.
A graph $G$ is called \emph{complete} if $E(G) = \{uv \mid u, v \in V(G)\}$.
A set of vertices $S \subseteq V(G)$ is a \emph{clique} in $G[S]$ is complete.
A complete graph on three vertices is called a \emph{triangle}.

A graph $G$ is called \emph{bipartite} if its vertex set can be partitioned into two
nonempty independent sets, which we will refer to as a \emph{bipartition} of $G$.

\paragraph*{Notation for Equivalence Relations.}
Let $\Omega$ be a set and $\sim$ an equivalence relation over $\Omega$.
For an element $x \in \Omega$ the \emph{equivalence class of $x$}, denoted by $[x]$,
is the set $\{y \in \Omega \mid x \sim y\}$.
We denote the set of all equivalence classes of $\sim$ by $\Omega/\sim$.

\paragraph*{Parameterized Complexity.} 
We give the basic definitions of parameterized complexity that are relevant to this work 
and refer to~\cite{CyganEtAl2015,DowneyFellows2013} for details.
Let $\Sigma$ be an alphabet. A \emph{parameterized problem} is a set $\Pi \subseteq \Sigma^* \times \bN$,
the second component being the \emph{parameter} which usually expresses a structural measure of the input.
A parameterized problem $\Pi$ is said to be \emph{fixed-parameter tractable}, or in the complexity class \FPT,
if there is an algorithm that for any $(x, k) \in \Sigma^* \times \bN$ 
correctly decides whether or not $(x, k) \in \Pi$, and runs in time $f(k) \cdot \card{x}^c$
for some computable function $f \colon \bN \to \bN$ and constant $c$.
We say that a parameterized problem is in the complexity class \XP, 
if there is an algorithm that for each $(x, k) \in \Sigma^* \times \bN$ correctly decides whether or not $(x, k) \in \Pi$,
and runs in time $f(k) \cdot \card{x}^{g(k)}$, for some computable functions $f$ and $g$.

The concept analogous to \NP-hardness in parameterized complexity is that of \W[1]-hardness, 
whose formal definition we omit.
The basic assumption is that $\FPT \neq \W[1]$, under which no \W[1]-hard problem admits an \FPT-algorithm.
For more details, see~\cite{CyganEtAl2015,DowneyFellows2013}.

\paragraph*{Strong Exponential Time Hypothesis.}
In 2001, Impagliazzo et al.\ conjectured that a brute force algorithm to solve the \textsc{$q$-SAT} problem
which given a CNF-formula with clauses of size at most $q$,
asks whether it has a satisfying assignment, is `essentially optimal.'
This conjecture is called the \emph{Strong Exponential Time Hypothesis},
and can be formally stated as follows.
(For a survey of conditional lower bounds based on \SETH and related conjectures, see~\cite{Vas15}.)
\begin{conjecture-nn}[$\SETH$, Impagliazzo et al.~\cite{ImpagliazzoPaturi2001,ImpagliazzoPaturiZane2001}]
	For every $\epsilon > 0$, there is a $q \in \calO(1)$ such that \textsc{$q$-SAT} on $n$
	variables cannot be solved in time $\Ostar((2-\epsilon)^n)$.
\end{conjecture-nn}

\subsection{Treewidth}
We now define the treewidth and pathwidth of a graph,
and later the notion of a nice tree decomposition that we will use later in this work.
\begin{definition}[Treewidth, Pathwidth]\label{def:tw}
	Let $G$ be a graph. A \emph{tree decomposition} of $G$ is a pair $(T, \calB)$ of a tree $T$ and
	an indexed family of vertex subsets 
	$\calB = \{B_t \subseteq V(G)\}_{t \in V(T)}$, called \emph{bags},
	satisfying the following properties.
	\begin{enumerate}[label={(T\arabic*)}, ref={T\arabic*}]
		\item\label{def:tw:vertex:coverage} $\bigcup_{t \in V(T)} B_t = V(G)$.
		\item\label{def:tw:edge:coverage} For each $uv \in E(G)$ there exists some $t \in V(T)$ such that $\{u, v\} \subseteq B_t$.
		\item\label{def:tw:subtrees} For each $v \in V(G)$, let $U_v \defeq \{t \in V(T) \mid v \in B_t\}$ be the nodes in $T$
			whose bags contain $v$. Then, $T[U_v]$ is connected.
	\end{enumerate}
	The \emph{width} of $(T, \calB)$ is $\max_{t \in V(T)} \card{B_t} - 1$, 
	and the \emph{tree-width} of a graph is the minimum width over all its tree decompositions.
	If $T$ is a path, then $(T, \calB)$ is called a \emph{path decomposition}, 
	and the \emph{path-width} of a graph is the minimum width over all its path decompositions.
\end{definition}
The following notion of a \emph{nice} tree decomposition 
allows for streamlining the description of dynamic programming algorithms over tree decompositions.
\begin{definition}[Nice Tree Decomposition]
	Let $G$ be a graph and $(T, \calB)$ a tree decomposition of $G$.
	Then, $(T, \calB)$ is called a \emph{nice tree decomposition}, if $T$ is rooted 
	and each node is of one of the following types.
	\begin{description}
		\item[Leaf.] A node $t \in V(T)$ is a \emph{leaf node}, if $t$ is a leaf of $T$ and $B_t = \emptyset$.
		\item[Introduce.] A node $t \in V(T)$ is an \emph{introduce node} if it has precisely one child $s$,
			and there is a unique vertex $v \in V(G) \setminus B_s$ 
			such that $B_t = B_s \cup \{v\}$. 
			In this case we say that \emph{$v$ is introduced at $t$}.
		\item[Forget.] A node $t \in V(T)$ is a \emph{forget node}, if it has precisely one child $s$,
			and there is a unique vertex $v \in B_s$ such that $B_t = B_s \setminus \{v\}$.
			In this case we say that \emph{$v$ is forgotten at $t$}.
		\item[Join.] A node $t \in V(T)$ is a \emph{join node}, if it has precisely two children $s_1$ and $s_2$, 
			and $B_t = B_{s_1} = B_{s_2}$.
	\end{description}
\end{definition}

It is known that any tree decomposition of a graph can be transformed in linear time
into a nice tree decomposition of the same width, with a relatively small number of bags.
\begin{lemma}[Kloks~\cite{Klo94}]\label{lem:nice:td}
	Let $G$ be a graph on $n$ vertices, and let $k$ be a positive integer.
	Any width-$k$ tree decomposition $(T, \calX)$ of $G$ of can be transformed in time $\calO(k \cdot \card{V(T)})$
	into a nice tree decomposition $(T', \calX')$ of width $k$ such that $\card{V(T')} = \calO(k\cdot n)$.
\end{lemma}

\subsection{Clique-Width, Branch Decompositions, and Module-Width}
We first define clique-width, introduced by Courcelle and Olariu~\cite{CourcelleOlariu2000}, 
and then the equivalent measure of \emph{module-width} that we will use in our algorithm.
We keep the definition of clique-width slightly informal 
and refer to~\cite{CourcelleOlariu2000} for more details.

Let $G$ be a graph. The \emph{clique-width} of $G$, denoted by $\cliquewidth(G)$,
is the minimum number of labels $\{1, \ldots, k\}$ 
needed to obtain $G$ using the following four operations:
\begin{enumerate}[label={\arabic*.}]
	\item Create a new graph consisting of a single vertex labeled $i$.
	\item Take the disjoint union of two labeled graphs $G_1$ and $G_2$.
	\item Add all edges between pairs of vertices of label $i$ and label $j$.
	\item Relabel every vertex labeled $i$ to label $j$.
\end{enumerate}

We now turn to the definition of module-width which is based on 
the notion of a rooted branch decomposition.
\begin{definition}[Branch decomposition]
	Let $G$ be a graph.
	A \emph{branch decomposition} of $G$ is a pair $(T, \decf)$ of a subcubic tree $T$ and a bijection $\decf \colon V(G) \to L(T)$.
	If $T$ is a caterpillar, then $(T, \decf)$ is called a \emph{linear branch decomposition}.
	If $T$ is rooted, then we call $(T, \decf)$ a \emph{rooted branch decomposition}.
	In this case, for $t \in V(T)$, we denote by $T_t$ the subtree of $T$ rooted at $t$, 
	and we define $V_t \defeq \{v \in V(G) \mid \decf(v) \in \leaves(T_t)\}$,
	$\overline{V_t} \defeq V(G) \setminus V_t$, and
	$G_t \defeq G[V_t]$.
\end{definition}

Module-width is attributed to Rao~\cite{Rao2006,Rao2008}.
On a high level, the module-width of a rooted branch decomposition bounds, at each of its nodes $t$,
the maximum number of subsets of $\overline{V_t}$ that make up the intersection of $\overline{V_t}$ 
with the neighborhood of some vertex in $V_t$.
\begin{definition}[Module-width]
	Let $G$ be a graph, and $(T, \decf)$ be a rooted branch decomposition of $G$. 
	For each $t \in V(T)$, let $\sim_t$ be the equivalence relation on $V_t$ defined as follows:
	\begin{align*}
		\forall u, v \in V_t \colon u \sim_t v \Leftrightarrow N_G(u) \cap \overline{V_t} = N_G(v) \cap \overline{V_t}
	\end{align*}
	
	The \emph{module-width} of $(T, \decf)$ is $\mw(T, \decf) \defeq \max_{t \in V(T)} \card{V_t/{\sim_t}}$.
	%
	The \emph{module-width of $G$}, denoted by $\modulew(G)$, 
	is the minimum module width over all rooted branch decompositions of $G$.
\end{definition}

We introduce some notation.  For a node $t \in V(T)$ and a set
$S \subseteq V(G_t)$, we let \(\eqc_t(S)\) be the set of all equivalence classes
of \(\sim_{t}\) which have a nonempty intersection with \(S\), and
\(\eqcbar_t(S)\) be the remaining equivalence classes of \(\sim_{t}\). Formally,
$\eqc_t(S) \defeq \{Q \in V_t/{\sim_t} \mid Q \cap S \neq \emptyset\}$ and
$\eqcbar_t(S) \defeq V_t/{\sim_t} \setminus \eqc_t(S)$.
Moreover, for a set of equivalence classes $\calQ \subseteq V_t/{\sim_t}$, we let $V(\calQ) \defeq \bigcup_{Q \in \calQ} Q$.

Let $(T, \decf)$ be a rooted branch decomposition of a graph $G$ and let $t \in V(T)$ be a node with children $r$ and $s$.
We now describe an operator associated with $t$ that tells us how the graph $G_t$ is formed from its subgraphs $G_r$ and $G_s$,
and how the equivalence classes of $\sim_t$ are formed from the equivalence classes of $\sim_r$ and $\sim_s$.
Concretely, we associate with $t$ a bipartite graph $\decaux_t$ on bipartition $(V_r/{\sim_r}, V_s/{\sim_s})$ such that:
\begin{enumerate}
	\item $E(G_t) = E(G_r) \cup E(G_s) \cup F$, where $F = \{uv \mid u \in V_r, v \in V_s, \{[u], [v]\} \in E(\decaux_t)\}$, and
	\item there is a partition $\calP = \{P_1, \ldots, P_h\}$ of $V(\decaux_t)$ such that $V_t/{\sim_t} = \{Q_1, \ldots, Q_h\}$,
		where for $1 \le i \le h$, $Q_i = \bigcup_{Q \in P_i} Q$.
		For each $1 \le i \le h$, we call $P_i$ the \emph{bubble} of the resulting equivalence class $\bigcup_{Q \in P_i} Q$ of $\sim_t$.
\end{enumerate}

As auxiliary structures, for $p \in \{r, s\}$, we let
$\bubblemap_p \colon V_p/{\sim_p} \to V_t/{\sim_t}$ be the map such that for all
$Q_p \in V_p/{\sim_p}$, $Q_p \subseteq \bubblemap_p(Q_p)$, i.e.\
$\bubblemap_p(Q_p)$ is the equivalence class of $\sim_t$ whose bubble contains
$Q_p$.
We call $(\decaux_t, \bubblemap_r, \bubblemap_s)$ the \emph{operator} of $t$.

\begin{theorem}[Rao, Thm.~6.6 in~\cite{Rao2006}]\label{thm:cw:mw}
	For any graph $G$, $\mw(G) \le \cw(G) \le 2 \cdot \mw(G)$,
	and given a decomposition of bounded clique-width, a decomposition of bounded module-width,
	and vice versa,
	can be constructed in time $\calO(n^2)$, where $n = \card{V(G)}$.
\end{theorem}

\subsection{Colorings}
Let $G$ be a graph.  An ordered partition $\calC = (C_1, \ldots, C_k)$ of $V(G)$
is called a \emph{coloring} of $G$ with $k$ colors, or a \emph{\(k\)-coloring}
of \(G\).  (Observe that for $i \in \{1, \ldots, k\}$, $C_i$ may be empty.)  For
$i \in \{1, \ldots, k\}$, we call $C_i$ the \emph{color class $i$}, and say that
the vertices in $C_i$ \emph{have color $i$}.  $\calC$ is called \emph{proper} if
for all $i \in \{1, \ldots, k\}$, $C_i$ is an independent set in $G$.

A coloring $\calC = (C_1, \ldots, C_k)$ of a graph $G$ 
is called a \emph{clique coloring (with $k$ colors)} if there
is no monochromatic maximal clique, i.e.\ no maximal clique $X$ in $G$ such that $X \subseteq C_i$ for some $i$.
In this work, we study the following computational problems.
\fancyproblemdef
	{\cliquecol}
	{Graph $G$, integer $k$}
	{Does $G$ have a clique coloring with $k$ colors?}
\vspace{-2em}
\fancyproblemdef
	{\qcliquecol for $q \ge 2$}
	{Graph $G$}
	{Does $G$ have a clique coloring with $q$ colors?}
	
The \textsc{$q$-Coloring} and \textsc{$q$-List Coloring} problems also make an appearance.
In the former, we are given a graph $G$ and the question is whether $G$ has a proper coloring with $q$ colors.
In the latter, we are additionally given a list $L(V) \subseteq \{1, \ldots, q\}$ for each vertex $v \in V(G)$,
and additionally require the color of each vertex to be from its list.

Whenever convenient, we alternatively denote a coloring of a graph with $k$ colors as a map
$\phi\colon V(G) \to \{1, \ldots, k\}$.
In this case, a restriction of $\phi$ to $S$ is the map $\phi|_S\colon S \to \{1, \ldots, k\}$
with $\phi|_S(v) = \phi(v)$ for all $v \in S$.
For any $T \subseteq V(G)$ with $S \subseteq T$, we say that $\phi|_T$ extends $\phi|_S$.
%

\section{Parameterized by Treewidth}\label{sec:treewidth}
In this section, we consider the \qcliquecol problem, for fixed $q \ge 2$, parameterized by treewidth.
First, in Section~\ref{sec:ccol:treewidth:algo}, 
we show that if we are given a tree decomposition of width $\treewidth$ of the input graph,
then \qcliquecol can be solved in time $\Ostar(q^{\treewidth})$.
After that, in Section~\ref{sec:ccol:treewidth:lb}, 
we show that this is tight according to \SETH,
by providing one reduction ruling out $\Ostar((2-\epsilon)^{\treewidth})$-time algorithms for \qqcliquecol{2} 
and another one ruling out $\Ostar((q-\epsilon)^{\treewidth})$-time algorithms for \qcliquecol when $q \ge 3$.

\subsection{Algorithm}\label{sec:ccol:treewidth:algo}
The algorithm is bottom-up dynamic programming along the given tree decomposition of the input graph.
As a subroutine, we will have to be able to check, at each bag $B_t$, if some subset $S \subseteq B_t$ 
contains a maximal clique in $G[B_t]$.
Doing this by brute force would add a \emph{multiplicative} factor of 
roughly~$2^{\treewidth}\cdot \treewidth^{\calO(1)}$ to the runtime which we cannot afford.
To avoid this increase in the runtime, we use fast subset convolution to build an oracle $\cliqueoracle_t$
that, once constructed, can tell us in constant time 
whether or not any subset $S \subseteq B_t$ contains a maximal clique in $G[B_t]$, for each node $t$.
Since it suffices to construct this oracle once per node,
this will infer only an \emph{additive} factor of $2^{\treewidth}\cdot\treewidth^{\calO(1)}$ per node to the runtime,
which does not increase the worst-case complexity for any~$q \ge 2$.

\begin{proposition}\label{prop:clique:oracle}
	There is an algorithm that given a graph $G$ on $n$ vertices, 
	constructs an oracle $\cliqueoracle_G$ in time $\Ostar(2^n)$,
	such that given a set $S \subseteq V(G)$, 
	$\cliqueoracle_G$ returns in constant time
	whether or not $S$ contains a clique that is maximal in $G$.
\end{proposition}
\begin{proof}
	Before we proceed with the proof, recall that for a set $\Omega$, and two functions $\alpha$ and $\beta$ defined on $2^{\Omega}$,
	their \emph{subset convolution} \(\convol\) is defined as: for all $S \in 2^{\Omega}$, 
	$(\alpha \convol \beta)(S) = \sum_{T \subseteq S}\alpha(T)\beta(S \setminus T)$.
	Let $f\colon 2^{V(G)} \to \{0, 1\}$ be the function defined as follows. For all $X \subseteq V(G)$, we let
	\begin{align*}
		f(X) \defeq\left\lbrace
			\begin{array}{ll}
				1, &\mbox{if $X$ contains a maximal clique,} \\
				0, &\mbox{otherwise.}
			\end{array}
		\right.
	\end{align*}
	
	To prove the statement, we have to show how to compute all values of $f$ within the claimed time bound.
	We define $g \colon 2^{V(G)} \to \{0, 1\}$ to be the function such that for all $X \subseteq V(G)$, 
	$g(X) = 1$ if and only if $X$ is a maximal clique in $G$. 
	The values of $g$ can be computed in time $\Ostar(2^n)$ by brute force.
	We define a function $c \colon 2^{V(G)} \to \{0, 1\}$ as $c(X) = 1$ for all $X \subseteq V(G)$,
	and we let $h = g \convol c$, which can be computed in time $\Ostar(2^n)$~\cite{BHKK07}.
	For each set $X$, we have that there are $h(X)$ subsets of $X$ that are a maximal clique in $G$.
	Finally, we obtain $f$ as 
	\begin{align*}
		\forall X \subseteq V(G)\colon f(X) \defeq\left\lbrace
			\begin{array}{ll}
				1, &\mbox{if } h(X) \ge 1, \\
				0, &\mbox{otherwise,}
			\end{array}
		\right.
	\end{align*}
	which costs an additional $\Ostar(2^n)$ in the runtime.
\end{proof}

\begin{theorem}\label{thm:ccol:alg:tw}
	For any fixed $q \ge 2$, 
	there is an algorithm that given an $n$-vertex graph $G$ and a tree decomposition of $G$ of width $\treewidth$,
	decides whether $G$ has a clique coloring with $q$ colors in 
	time~$\calO(q^{\treewidth} \cdot \treewidth^{\calO(1)} \cdot n)$,
	and constructs one such coloring, if it exists.
\end{theorem}
\begin{proof}
	First, we transform the given tree decomposition of $G$ into a nice tree deocmposition $(T, \calB)$.
	This can be done in $\calO(n)$ time by Lemma~\ref{lem:nice:td}.
	We may assume that the bags at leaf nodes are empty, and that $T$ is rooted in some node $\rootnode \in V(T)$, 
	and $B_{\rootnode} = \emptyset$.
	
	We do standard bottom-up dynamic programming along $T$.
	Let $t \in V(T)$. 
	A partial solution is a $q$-coloring of $G_t$ that satisfies one additional property.
	Suppose that in some coloring of $G_t$, there is a monochromatic maximal clique $X$ in $G_t$
	that has some vertex $v \in V_t \setminus B_t$.
	Then, $v$ has no neighbors in $V(G) \setminus V_t$, therefore $X$ is also a maximal clique in $G$.
	This means that the present coloring cannot be extended to a coloring 
	in which $X$ becomes non-maximal, and therefore we can disregard it.
	
	In light of this, we define the table entries as follows.
	For each $t \in V(T)$ and function $\coloring_t \colon B_t \to \{1, \ldots, q\}$, 
	we let $\dptable[t, \coloring_t] = 1$ 
	if and only if there is a $q$-coloring $\coloring$ of $G_t$ such that
	\begin{itemize}
		\item $\coloring|_{B_t} = \coloring_t$, and
		\item for each clique $X$ in $G_t$ that is monochromatic under $\coloring$, $X \subseteq B_t$.
	\end{itemize}
	
	Since $B_{\rootnode} = \emptyset$, we can immediately observe that the solution to the instance
	can be read off the table entries at the root node, once computed. 
	Throughout the following we denote by $\gamma_\emptyset$ the $q$-coloring defined on an empty domain.
	\begin{nestedobservation}
		$G$ has a clique coloring with $q$ colors if and only if $\dptable[\rootnode, \gamma_\emptyset] = 1$.
	\end{nestedobservation}
	
	We now show how to compute the table entries for the different types of nodes, 
	assuming that the table entries at the children, if any, have previously been computed.
	\begin{description}
		\item[Leaf Node.]
		If $t$ is a leaf node, then $B_t = \emptyset$ and we only have to consider the empty coloring.
		We set $\dptable[t, \gamma_\emptyset] = 1$.
		\item[Introduce Node.]
		Let $t \in V(T)$ be an introduce node with child $s$, and let $v$ be the vertex introduced at $t$,
		i.e.\ we have that $B_t = B_s \cup \{v\}$. 
		Since $V_t \setminus B_t = (V_t \setminus \{v\}) \setminus (B_t \setminus \{v\}) = V_s \setminus B_s$,
		and since $v$ has no neighbors in $V_t \setminus B_t$ by the properties of a tree decomposition,
		it is clear that a coloring of $G_t$ has a monochromatic maximal clique with a vertex in $V_t \setminus B_t$
		if and only if
		its restriction to $V_s$ is a coloring of $G_s$ that has a monochromatic maximal clique with a vertex in $V_s \setminus B_s$.
		Therefore, for each $\coloring_t \colon B_t \to \{1, \ldots, q\}$, we simply let $\dptable[t, \coloring_t] = 1$ 
		if and only if $\dptable[s, \coloring_t|_{B_s}] = 1$.
		\item[Join Node.]
		Let $t \in V(T)$ be a join node with children $s_1$ and $s_2$ and recall that $B_t = B_{s_1} = B_{s_2}$.
		In this case, for any $\coloring_t\colon B_t \to \{1, \ldots, q\}$,
		$G_t$ has a $q$-coloring $\coloring$ with $\coloring|_{B_t} = \coloring_t$ 
		without a monochromatic maximal clique in $V_t \setminus B_t$
		if and only if
		the analogous condition holds for both $G_{s_1}$ and $G_{s_2}$.
		Therefore, for all such $\coloring_t$, we let $\dptable[t, \coloring_t] = 1$ 
		if and only if $\dptable[s_1, \coloring_t] = \dptable[s_2, \coloring_t] = 1$.
		\item[Forget Node.]
		Let $t \in V(T)$ be a forget node with child $s$ and let $v$ be the vertex forgotten at $t$,
		i.e.\ $B_s = B_t \cup \{v\}$.
		A partial solution at node $s$ may have a monochromatic maximal clique using the vertex $v$,
		provided that the clique is fully contained in $B_s$,
		while partial solutions at the node $t$ may not.
		Therefore, for a given coloring $\coloring_t\colon B_t \to \{1, \ldots, q\}$,
		we can check
		whether or not there is a partial solution in $G_t$ whose restriction to $B_t$ is equal to $\coloring_t$ as follows.
		For each color $c \in \{1, \ldots, q\}$, extend $\coloring_t$ to a coloring $\coloring_s$ of $B_s$ by assigning vertex $v$ color $c$.
		Check if there is a partial solution at node $s$ whose restriction to $B_s$ is $\coloring_s$,
		and if there is no maximal clique in $N(v) \cap \coloring_s^{-1}(c)$.
		If this is the case for some color $c$, then we set $\dptable[t, \coloring_t] = 1$, 
		and if it is not the case for any color, then we set $\dptable[t, \coloring_t] = 0$.
		It is clear from this description that this is correct.
		However, we have to apply one additional trick to ensure that we do not exceed the targeted runtime bound.
		Since in the worst case, there are $q^{\treewidth}$ many colorings to consider 
		(note that since $t$ is a forget node, we have that $\card{B_t} \le \treewidth$),
		we can only spend constant time for the computation of each entry $\dptable[t, \coloring_t]$.
		Verifying if $N(v) \cap \coloring_s^{-1}(v)$ contains a maximal clique 
		may take time $\Ostar(2^{\treewidth})$ in the worst case.
		We overcome this issue by constructing a maximal clique containment oracle $\cliqueoracle_t$ of $G[B_t]$
		using Proposition~\ref{prop:clique:oracle}.
		Once constructed, this oracle allows for checking whether a set contains a maximal clique in constant time.
		We describe the entire procedure of how to compute table entries at forget nodes 
		in Algorithm~\ref{alg:ccol:tw:forget}.
	\end{description}
		\begin{algorithm}
			\SetKwInOut{Input}{Input}
			\Input{$G$, $(T, \calB)$ as above, forget node $t \in V(T)$}
			Let $v \in B_s \setminus B_t$ be the vertex forgotten at $t$\;
			Construct the clique oracle $\cliqueoracle_t$ of $G[B_t]$ using Proposition~\ref{prop:clique:oracle}\;
			\ForEach{$\coloring_t \colon B_t \to \{1, \ldots, q\}$}{
				$\dptable[t, \coloring_t] \gets 0$\;
				\ForEach{$c \in \{1, \ldots, q\}$}{
					Let $\coloring_s \colon B_s \to \{1, \ldots, q\}$ be such that
						for all $u \in B_t$, $\coloring_s(u) = \coloring_t(u)$, and
						$\coloring_s(v) = c$\;
					\If{$\dptable[s, \coloring_s] = 1$}{
						\lIf{$\cliqueoracle_t(N(v) \cap \coloring_t^{-1}(c)) = 0$}{$\dptable[t, \coloring_t] \gets 1$}
					}
				}
			}
			\caption{Algorithm to compute all table entries at a forget node $t$ with child $s$,
				assuming all table entries at $s$ have been computed. 
				(Notation: For a set $S \subseteq B_t$, $\cliqueoracle_t(S) = 0$ if and only $G[S]$ contains no maximal clique.)}
			\label{alg:ccol:tw:forget}
		\end{algorithm}

	This completes the description of the algorithm. 
	Correctness follows from the description of the computation of the table entries, by induction on the height of each node.
	For the runtime, observe that there are at most $q^{\treewidth + 1}$ table entries to consider at each node,
	and it is clear that the computation of a table entry at a leaf, introduce, or join node takes constant time.
	For forget nodes, the construction of $\cliqueoracle_t$ takes 
	time~$\calO(2^{\treewidth} \cdot \treewidth^{\calO(1)}) \le \calO(q^{\treewidth} \cdot \treewidth^{\calO(1)})$
	by Proposition~\ref{prop:clique:oracle}.
	After that, computation of each table entry takes $\calO(q) = \calO(1)$ time.
	Therefore, the total time to compute the table entries at a forget node is $\calO(q^{\treewidth} \cdot \treewidth^{\calO(1)})$.
	Since $\card{V(T)} = \calO(\treewidth \cdot n)$, 
	the total runtime of the algorithm is~$\calO(q^{\treewidth} \cdot \treewidth^{\calO(1)} \cdot n)$.
	Using memoization techniques, the algorithm can construct a coloring, if one exists.
\end{proof}

\subsection{Lower Bound}\label{sec:ccol:treewidth:lb}
In this section we show that the previously presented algorithm is optimal under \SETH.
In fact, we prove hardness for a much larger parameter, namely the distance to a linear forest (for $q = 2$),
and the distance to a caterpillar forest (for $q \ge 3$).
Note that both paths and caterpillars have pathwidth $1$, and clearly, they do not contain any cycles.
Therefore, a lower bound parameterized by the (vertex deletion) distance to a linear/caterpillar forest 
implies a lower bound for the parameter pathwidth plus feedback vertex set number.
For $q = 2$, we give a reduction from \textsc{$s$-Not-All-Equal SAT} ($s$-\naesat) on $n$ variables.
Cygan et al.~\cite{CyganEtAl2016} showed that under \SETH, 
for any $\epsilon > 0$, there is some constant $s$ such that $s$-\naesat cannot be solved in time $\Ostar((2-\epsilon)^n)$.
For all $q \ge 3$, we reduce from \textsc{$q$-List Coloring}, 
where we are given a graph $G$ and a list for each of its vertices
which is a subset of $\{1, 2, \ldots, q\}$, and the question is whether $G$ has
a proper coloring such that each vertex receives a color from its list.
Parameterized by the size $t$ of a deletion set to a linear forest,
this problem is known to have no $\Ostar((q-\epsilon)^t)$-time algorithms under \SETH~\cite{JaffkeJansen2017}.
Our construction uses the fact that on triangle-free graphs, 
the proper colorings and the clique colorings coincide,
and exploits properties of Mycielski graphs.

We first give the lower bound for the case $q = 2$.
We would like to remark that Kratochv\'{i}l and Tuza~\cite{KT2002} gave a reduction from \textsc{Not-All-Equal SAT} to 
\qqcliquecol{2} as well, but their reduction does not imply the fine-grained lower bound we aim for here:
the resulting graph is at distance $2n$ to a disjoint union of cliques of constant size (at most $s$).
This only rules out $\Ostar((\sqrt{2}-\epsilon)^t)$-time algorithms parameterized by pathwidth,
and does not give any lower bound if the feedback vertex set number is another component of the parameter.
\begin{theorem}
	For any $\epsilon > 0$, \qqcliquecol{2} parameterized by the distance $t$ to a linear forest 
	cannot be solved in time $\Ostar((2-\epsilon)^t)$,
	unless \SETH fails.
\end{theorem}
\begin{proof}
We give a reduction from the well-known $s$-\naesat problem, in which we are given a boolean CNF formula $\cnff$
whose clauses are of size at most $s$, and the question is whether there is a truth assignment to the variables 
of $\cnff$, such that in each clause, at least one literal evaluates to true and at least one literal evaluates to false.

Let $\cnff$ be a boolean CNF formula on $n$ variables $x_1,\ldots, x_n$ with maximum clause size $s$. We denote by $\clauses(\cnff)$ the set of clauses of $\cnff$ and by $\variables(C)$ the set of variables that appear in the clause $C$ of $\cnff$.

Given $\phi$, we construct an instance $G_\phi$ for \qqcliquecol{2} as follows. For each variable $x_i$, we create a vertex $v_i$ in $G$. Let $V'=\{v_1,\ldots,v_n\}$. For each set $S$ of variables, let $V_S=\{v_i~|~x_i\in S\}$. For each clause $C_i$ of $\cnff$, we add the following clause gadget to $G_\phi$. If $C_i$ is monotone, add a path on four vertices to $G_\phi$, the end vertices of which are $a_i$ and $b_i$. Make $N(a_i)\cap V'=N(b_i)\cap V'=V_{\variables(C_i)}$, and make $V_{\variables(C_i)}\subset V'$ a clique. If $C_i$ is not monotone, let $\positive(C)$ (resp.\ $\negative(C)$) denote the set of variables with positive (resp.\ negative) literals in $C$. Add a path on three vertices to $G_\phi$, the end vertices of which are $a_i$ and $b_i$, make $N(a_i)\cap V'=V_{\positive(C)}$ and make $V_{\positive(C)}$ a clique. Analogously, make $N(b_i)\cap V'=V_{\negative(C)}$ and make $V_{\negative(C)}$ a clique. Finally, add two adjacent vertices $u,v$ to $G_\phi$ and make $N[u]=N[v]=\{u,v\}\cup V'$. See Figure~\ref{fig:reduction1}.

\begin{figure}
\centering
\includegraphics[width=8cm]{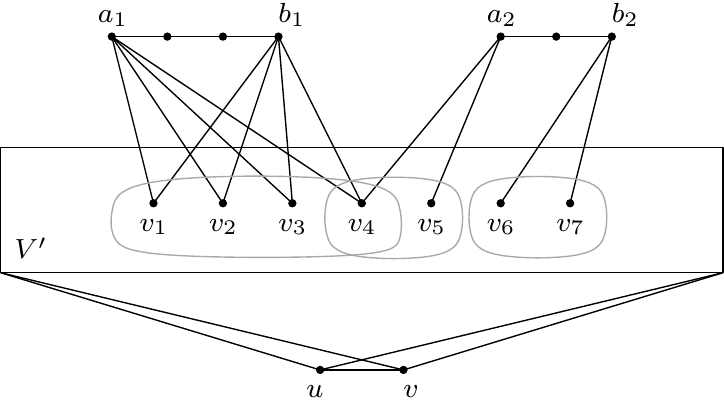}
\caption{Depiction of $G_\phi$ with two clauses, namely a monotone clause $C_1=\neg x_1\vee \neg x_2 \vee \neg x_3 \vee \neg x_4$ and a non-monotone clause $C_2=x_4 \vee x_5 \vee \neg x_6\vee\neg x_7$. Note that $G_\phi - V'$ is a linear forest.}
\label{fig:reduction1}
\end{figure}

We will show that $G_\phi$ is a yes-instance to \qqcliquecol{2} if and only if $\cnff$ is a yes-instance to $s$-\naesat. We first make the following observation about the maximal cliques of $G_\phi$, which follows directly from the fact that the vertices $u$ and $v$ are complete to $V'$. 

\begin{nestedobservation}\label{obs:maxcliquesV}
The vertices $u$ and $v$ belong to every maximal clique of $G_\phi[V'\cup\{u,v\}]$.
\end{nestedobservation}

\begin{nestedclaim}\label{claim:aibi}
Let $f: V(G_\phi) \rightarrow \{0,1\}$ be a 2-clique coloring of $G_\phi$ and $C_i$ be a clause of $\phi$. Then, if $C_i$ is monotone, then $f(a_i)\neq f(b_i)$. Otherwise, $f(a_i)=f(b_i)$.
\end{nestedclaim}

\begin{claimproof}
If $C_i$ is monotone, $a_i$ and $b_i$ are the end vertices of a path on four vertices, each edge of which is a maximal clique of $G_\phi$. Thus, $f(a_i)\neq f(b_i)$ in any 2-clique coloring $f$ of $G_\phi$. Similarly, if $C_i$ is not monotone, $a_i$ and $b_i$ are the end vertices of a path on three vertices, each edge of which is a maximal clique of $G_\phi$. Hence $f(a_i)=f(b_i)$.
\end{claimproof}

Now, suppose $G$ has a 2-clique coloring $f: V(G_\phi) \rightarrow \{0,1\}$. We construct a truth assignment for $\{x_1,\ldots,x_n\}$ according to the colors assigned to the vertices of $V'$ by $f$. That is, if $f(v_i)=0$, we set $x_i$ to false, and if $f(v_i)=1$, we set $x_i$ to true. We will now show that this assignment satisfies all clauses of $\phi$. Let $C_i$ be a clause of $\phi$. First, assume that $C_i$ is monotone. By Claim~\ref{claim:aibi}, $f(a_i)\neq f(b_i)$. Since $V_{\variables(C_i)}\cup\{a_i\}$ is a maximal clique of $G_\phi$, the vertices of $V_{\variables(C_i)}$ cannot all be colored with $f(a_i)$. Similarly, $V_{\variables(C_i)}\cup\{b_i\}$ is a maximal clique of $G_\phi$, the vertices of $V_{\variables(C_i)}$ cannot all be colored with $f(b_i)$. Thus, there exist two vertices $v_j,v_k\in V_{\variables(C_i)}$ such that $f(v_j)\neq f(v_k)$. Since $C_i$ is monotone, this implies that $x_j$ and $x_k$ are not both evaluated to the same value and therefore $C_i$ is satisfied. Now assume $C_i$ is not monotone. By Claim~\ref{claim:aibi}, $f(a_i)=f(b_i)$. Hence, since $V_{\positive(C_i)}\cup\{a_i\}$ and $V_{\negative(C_i)}\cup\{b_i\}$ are maximal cliques of $G$, there exists $v_j\in V_{\positive(C_i)}$ and $v_k\in V_{\negative(C_i)}$ such that $f(v_j)=f(v_k)$. This implies that $x_j$ and $x_k$ are not evaluated to the same value under the proposed assignment and thus $C_i$ is satisfied.

For the other direction, assume $\phi$ admits an assignment $\xi$ satisfying all clauses. We construct a clique coloring $f: V(G_\phi) \rightarrow \{0,1\}$ for $G_\phi$ in the following way. Color the vertices of $V'$ according to the assignment of the variables of $\phi$. That is, if $\xi(x_i)=\mbox{true}$  (resp.\ $\xi(x_i)=\mbox{false}$), define $f(v_i)=1$ (resp.\ $f(v_i)=0$). If $C_i$ is monotone, let $a_ia_i'b_i'b_i$ be the path on four vertices connecting $a_i$ and $b_i$ in the clause gadget of $C_i$. Define $f(a_i)=f(b_i')=1$ and $f(a_i')=f(b_i)=0$. If $C_i$ is not monotone, let $a_ia_i'b_i$ be the three vertex path connecting $a_i$ and $b_i$ in the clause gadget of $C_i$. If all the vertices of either $V_{\positive(C_i)}$ or $V_{\negative(C_i)}$ are colored 1, set $f(a_i)=f(b_i)=0$ and $f(a_i')=1$. Otherwise set $f(a_i)=f(b_i)=1$ and $f(a_i')=0$. Finally, define $f(u)=0$ and $f(v)=1$. To see that this is indeed a 2-clique coloring of $G_\phi$, first note that by Observation~\ref{obs:maxcliquesV}, no maximal clique contained in $G_\phi[V'\cup\{u,v\}]$ is monochromatic. Furthermore, since all paths of the clause gadgets are properly colored, no maximal clique contained in $G_\phi- (V'\cup\{u,v\})$ is monochromatic. It remains to show that for each clause $C_i$, the maximal cliques defined by $N[a_i]$ and $N[b_i]$ are not monochromatic. Let $C_i$ be a monotone clause. Since $C_i$ is satisfied, there exist $x_j,x_k\in \variables(C_i)$ such that $\xi(x_j)\neq \xi(x_k)$. Hence, $f(v_j)\neq f(v_k)$, which shows that $N[a_i]$ and $N[b_i]$ are each not monochromatic. If $C_i$ is not monotone, by definition the vertices of $N[a_i]$ and $N[b_i]$ are not all colored 1. Suppose all the vertices of $N[a_i]$ are colored 0. In particular, we have $f(a_i)=f(b_i)=0$. This implies that, by construction, all the vertices of $N(b_i)=V_{\negative(C_i)}$ are colored 1. However, this is a contradiction with the fact that the clause $C_i$ is satisfied, since all its literals are evaluated to false. Hence, $f$ is indeed a 2-clique coloring of $G_\phi$.

Finally, note that $G-V'$ is a disjoint union of paths of length at most four. 
Hence, $G$ is at distance $n$ to a linear forest.
Therefore, if for some $\epsilon > 0$, \qqcliquecol{2} parameterized by the distance $t$ to a linear forest can be solved in 
	time $\Ostar((2-\epsilon)^t)$, then $s$-\naesat can be solved in time $\Ostar((2-\epsilon)^n)$, 
	which would contradict \SETH~\cite{CyganEtAl2016}. This concludes the proof.
\end{proof}

We now turn to the case $q \ge 3$.
Our reduction is from \textsc{$q$-List-Coloring} parameterized by the distance $t$ to a linear forest,
which has no $\Ostar((q-\epsilon)^t)$-time algorithms under \SETH 
by a theorem due to Jaffke and Jansen~\cite{JaffkeJansen2017}.
For technical reasons, we need the lower bound in a slightly stronger form, 
in particular it has to hold when the input graphs are triangle-free.
The reduction presented in~\cite{JaffkeJansen2017} is from $s$-SAT on $n$ variables, 
and given a formula $\phi$, the graph $G_\phi$ of the resulting \qlistcol instance has the following structure.
The truth assignments of the variables of $\phi$ are encoded as colorings of 
a set of vertices $V'$ that are independent in $G_\phi$,
and for each clause $C$ in $\phi$ and each coloring of some subset $V_C \subseteq V'$ 
that corresponds to a truth assignment $\mu$ that does not satisfy $C$,
there is a path $P_\mu$ in $G$ that cannot be properly list colored 
if and only if the coloring $\mu$ appears on $V_C$.
This is ensured by connecting $P_\mu$ to $V_C$ via a matching,
which does not introduce triangles.
Since each edge of $G_\phi$ is either on such a path or part of one of such matching, 
there are no triangles in~$G_\phi$.
\begin{theorem}[Jaffke and Jansen~\cite{JaffkeJansen2017}]\label{thm:qlistcolbound}
For any $\epsilon > 0$ and any fixed $q\geq 3$, $q$-{\sc List Coloring} on triangle-free graphs parameterized by the distance $t$ to a linear forest cannot be solved in time $\Ostar((q-\epsilon)^t)$,
	unless \SETH fails.
\end{theorem}

\begin{theorem}\label{thm:cliquecolgeq3}
	For any $\epsilon > 0$ and any fixed $q\geq 3$, $q$-{\sc Clique Coloring} parameterized by the distance $t$ to a caterpillar forest
	cannot be solved in time $\Ostar((q-\epsilon)^t)$,
	unless \SETH fails. 
\end{theorem}
\begin{proof}
  We give a reduction from $q$-{\sc List Coloring} on triangle-free graphs
  parameterized by distance to linear forest. In this proof we use the phrases
  ``\(q\)-colorable'' as short for ``can be properly colored with at most \(q\)
  colors'', and ``\(q\)-coloring'' as short for ``a proper coloring with at most
  \(q\) colors''.  To construct our instance of \qqcliquecol{q}, we will first
  describe the construction of a color selection gadget, and then describe how
  this gadget is attached to rest of the graph. The description of the color
  selection gadget makes use of the famous Mycielski graphs. For completeness,
  we briefly describe how Mycielski graphs are recursively constructed and some
  of their useful properties. For every $p\geq 2$, the Mycielski graph $M_p$ is
  a triangle-free graph with chromatic number $p$. For $p=2$, we define
  $M_2=K_2$. For $p\geq 3$, the graph $M_p$ is obtained from $M_{p-1}$ as
  follows. Let $V(M_{p-1})=\{v_1,\ldots,v_n\}$. Then
  $V(M_p)=V(M_{p-1})\cup\{u_1,\ldots,u_n,w\}$. The vertices of $V(M_{p-1})$
  induce a copy of $M_{p-1}$ in $M_p$, each $u_i$ is adjacent to all the
  neighbors of $v_i$ in $M_{p-1}$ and $N(w)=\{u_1,...,u_n\}$. Hence,
  $|V(M_p)|=3\cdot 2^{p-2}-1$. Moreover, it is known that $M_p$ is
  edge-critical, that is, the deletion of any edge of $M_p$ leads to a
  $(p-1)$-colorable graph (see for instance \cite{BondyMurty,Lovasz}). For our
  construction, we will use the graph $M_p'$, obtained from $M_p$ by the
  deletion of an arbitrary edge $xy$. The following observation follows directly
  from the fact that $M_p$ is edge-critical.

\begin{nestedobservation}\label{obs:samecolor}
Let $M_p'$ be the graph obtained from $M_p$ by the deletion of an edge $xy$. Then, $M_p'$ is $(p-1)$-colorable, and in any $(p-1)$-coloring of $M_p'$, the vertices $x$ and $y$ receive the same color.
\end{nestedobservation}

\noindent{\it Color selection gadget.} 
We construct a gadget $H_q$ in the following way. Consider $q$ disjoint copies of $M_{q+1}'$. For $1\leq i\leq q$, let $x_iy_i$ be the edge removed from $M_{q+1}$ in order to obtain the $i$th copy of $M_{q+1}'$. For each $i$, add $q-1$ false twins to $y_i$. We denote these vertices by $y_{ij}$, with $1\leq j\leq q$, $j\neq i$. Then delete the vertex $y_i$, for every $i$. Note that this graph is still $q$ colorable and, by Observation~\ref{obs:samecolor}, in every such $q$-coloring, for each $i$, the vertices $x_i$ and $y_{ij}$, for all $j\neq i$, receive the same color. Now we add $q\choose 2$ edges to connect the copies of $M_{q+1}'$: for $1\leq i<j\leq q$, add the edge $y_{ij}y_{ji}$ to $H_q$. Note that $H_q$ remains triangle-free after the addition of these edges, since for all $1\leq i<j\leq q$, $N(y_{ij})\cap N(y_{ji})=\emptyset$. We will need the following property of the $q$-colorings of $H_q$.

\begin{nestedclaim}\label{claim:colorxi}
The graph $H_q$ is $q$-colorable. Moreover, in any $q$-coloring $\phi$ of $H_q$, $\phi(x_i)\neq\phi(x_j)$ for all $1\leq i<j\leq q$.
\end{nestedclaim}

\begin{claimproof}
Suppose for a contradiction that there exists a $q$-coloring of $H_q$ such that $\phi(x_i)=\phi(x_j)$, for some $i\neq j$. By Observation~\ref{obs:samecolor}, we know that $\phi(x_i)=\phi(y_{ij})$. Similarly, $\phi(x_j)=\phi(y_{ji})$. This implies that $\phi(y_{ij})=\phi(y_{ji})$, which is a contradiction, since $y_{ij}$ and $y_{ji}$ are adjacent by construction. To see that a $q$-coloring indeed exists for $H_q$, first note that, by Observation~\ref{obs:samecolor}, each copy of $M_{q+1}'$ has a $q$-coloring in which $x_i$ and $y_i$ are assigned the same color. We can then permute the colors within a copy to obtain a proper coloring of that copy in which $x_i$ and $y_i$ receive color $i$. To complete the coloring, assign color $i$ to every $y_{ij}$ that is a false twin of $y_i$. This yields a proper $q$-coloring of $H_q$.
\end{claimproof}

We are now ready to describe the construction of our instance $G'$ to $q$-{\sc Clique Coloring}. Let $(G,L)$ be an instance of $q$-{\sc List Coloring} on triangle-free graphs that is at distance $k$ from a linear forest. We construct $G'$ as follows. Add a copy of $G$ and a copy of $H_q$ to $G'$. We denote by $V'$ the set of vertices corresponding to $V(G)$ in $G'$. For each $v\in V'$, add $q-|L(v)|$ vertices adjacent to $v$. We denote these vertices by $\{v_j~|~j\notin L(v)\}$. Finally, make $v_j$ adjacent to all the vertices of $\{x_\ell~|~\ell\neq j\}$. See Figure~\ref{fig:reduction2}.

\begin{figure}
\centering
\includegraphics[width=11cm]{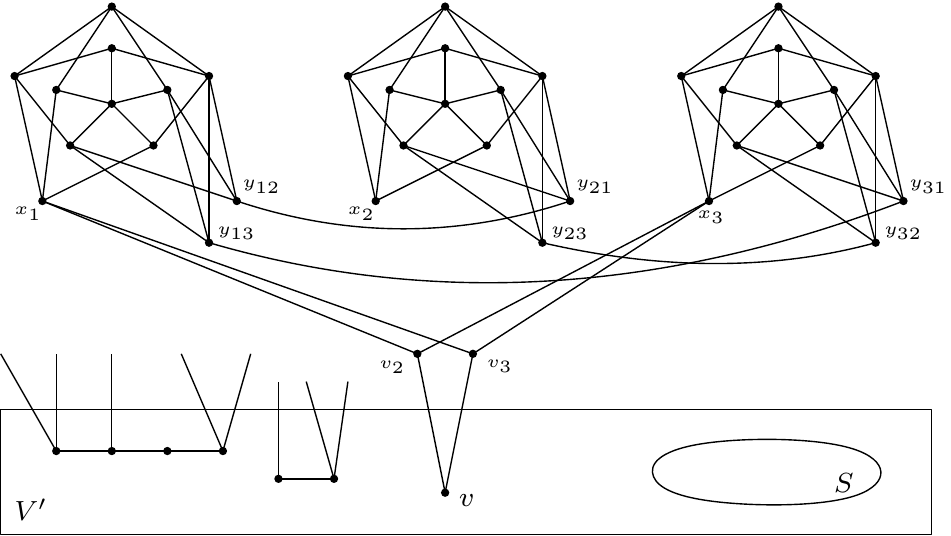}
\caption{In this instance, $q=3$ and $L(v)=\{1\}$. Note that $G'-(S\cup V(H_q))$ is a caterpillar forest.}
\label{fig:reduction2}
\end{figure}

Note that $G'$ is a triangle-free since $H_q$ and $G$ are triangle free, and $N(v_j)\cap V'=\{v\}$ and $N(v_j)\cap V(H_q)$ is an independent set. Furthermore, let $S\subseteq V(G)$ be a set such that $G-S$ is a linear forest and $|S|=t$. Then $S\cup V(H_q)$ is such that each connected component of $G'-(S\cup V(H_q))$ is a caterpillar and $|S\cup V(H_q)|=t+q(3\cdot 2^{q-1}+q-3)=t+\calO(1)$, since $q$ is a constant.

We will show that $(G,L)$ is a yes-instance to $q$-{\sc List Coloring} if and only if $G'$ is a yes-instance to $q$-{\sc Clique Coloring}. Note that since $G'$ is a triangle-free graph, every clique coloring of $G'$ is actually a proper coloring of it also. First, suppose $(G,L)$ is a yes-instance to $q$-{\sc List Coloring} and let $\phi$ be a $q$-list coloring for $G$.   We give a $q$-coloring $\phi'$ for $G'$ in the following way. If $v\in V'$, make $\phi'(v)=\phi(v)$. For each $v_j\in N(v)$, make $\phi'(v_j)=j$. Note that since $j\notin L(v)$, we have that $\phi'(v)\neq\phi(v_j)$. Finally, consider a proper $q$-coloring of $H_q$. By Claim~\ref{claim:colorxi}, the vertices $x_1,\ldots,x_q$ were assigned pairwise distinct colors. Without loss of generality, we can assume $x_i$ received color~$i$. Extend $\phi$ to the remaining vertices of $G'$ according to this coloring of $H_q$. This leads to proper $q$-coloring of $G'$, since $\phi(v_j)=j$ and $v_j$ is not adjacent to $x_j$.

Now assume $G'$ admits a $q$-clique coloring $\phi$. We will show that $\phi|_{V'}$ is a $q$-list coloring for $(G,L)$. Since $G'$ is triangle-free, it is clear that $\phi|_{V'}$ is a proper coloring of $G$. It remains to show it satisfies the constraints imposed by the lists. By Claim~\ref{claim:colorxi}, we can again assume that $\phi(x_i)=i$, for every $i$. For every $v\in V'$, since $\{x_\ell~|~\ell\neq j\}\subset N(v_j)$, we necessarily have $\phi(v_j)=j$. Finally, since for every $c\notin L(v)$ there is a neighbor of $v$ that is colored $c$ (namely $v_c$), we conclude that $\phi(v)\in L(v)$.

Now, suppose that \qcliquecol admits an algorithm running in time $\Ostar((q-\epsilon)^{t'})$, for some $\epsilon > 0$,
where $t'$ is the distance of the input graph to a caterpillar forest.
Then, we can solve \textsc{$q$-List-Coloring} paramterized by the distance $t$ to a linear forest
by applying the above reduction, giving a \qcliquecol instance at distance $t + \calO(1)$ to a caterpillar forest, 
and solving the resulting \qcliquecol instance.
Correctness is argued in the previous paragraphs, and the runtime of the resulting algorithm is
$\Ostar((q-\epsilon)^{t + \calO(1)}) = \Ostar((q-\epsilon)^t)$, 
contradicting \SETH by Theorem~\ref{thm:qlistcolbound}.
\end{proof}

Since the instance of \qcliquecol constructed in the proof of Theorem~\ref{thm:cliquecolgeq3} is a triangle-free graph, we obtain the following corollary.

\begin{corollary}
	For any $\epsilon > 0$ and any fixed $q\geq 3$, $q$-{\sc Coloring} on triangle-free graphs parameterized by the distance $t$ to a caterpillar forest
	cannot be solved in time $\Ostar((q-\epsilon)^t)$,
	unless \SETH fails. 
\end{corollary}

\section{Parameterized by Clique-width}\label{sec:cliquewidth}
In this section, we give an $\XP$-time algorithm for \cliquecol parameterized by clique-width,
more precisely, parameterized by the equivalent measure module-width.
We provide an algorithm that given an $n$-vertex graph $G$ with one of its rooted branch decompositions of module-width $w$
and an integer $k$, decides whether $G$ has a clique coloring with $k$ colors in time $n^{f(w)}$, 
where $f(w) = 2^{2^{\calO(w)}}$.
Before we describe the algorithm, we give a high level outline of its main ideas, 
and where the double exponential dependence on $w$ in the degree of the polynomial comes from.

The algorithm is bottom-up dynamic programming along the given branch decomposition of the input graph.
Let $t$ be some node in the branch decomposition.
To keep the number of table entries bounded by something that is \XP in the module-width,
we have to find a way to group color classes into a number of \emph{types} that is upper bounded
by a function of $w$ alone.
The intention is that two color classes of the same type 
are interchangeable with respect to the underlying coloring being 
completable to a valid clique coloring of the whole graph.
Partial solutions (colorings of the subgraph $G_t$) can then be described by remembering,
for each type, how many color classes of that type there are.
If the number of types is $f(w)$ for some function $f$, 
this gives an upper bound of $n^{f(w)}$ on the number of table entries at each node of the branch decomposition.

Let us discuss what kind of information goes into the definition of a type.
Since the final coloring of $G$ has to avoid monochromatic maximal cliques,
we maintain information about cliques in $G_t$ that are or may become 
monochromatic maximal cliques in some extension of the coloring at hand.
A natural attempt would be to consider and describe maximal cliques in $G_t$ 
by their intersection patterns with the equivalence classes of $\sim_t$.
However, it is not sufficient to consider only maximal cliques in $G_t$;
given a maximal clique $X$ in $G_t$,
it may happen that in $\overline{V_t}$ there is a vertex $v$ 
that is adjacent to a strict subset $Y \subset X$ of that clique,
forming a maximal clique with $Y$ -- which does not fully contain $X$ --
in a supergraph of $G_t$.
Considering the equivalence classes of $\sim_t$,
this implies that the equivalence classes containing $Y$ and the ones containing $X \subset Y$ are disjoint.
We therefore consider cliques $X$ that are 
\emph{maximal in the subgraph induced by the equivalence classes containing vertices of $X$}.
We call such cliques $X$ \emph{\eqcmax}, 
and observe that with a little extra information,
we can keep track of the forming and disintegrating of \eqcmax cliques along the branch decomposition.
If an \eqcmax clique is fully contained in some set of vertices (/color class) $C$, 
then we call it \emph{potentially bad for $C$}.
A potentially bad clique is described via its \emph{profile}, which consists of the intersection pattern
with the equivalence classes of $\sim_t$, and some extra information.
At each node, there are at most $2^{\calO(w)}$ profiles.

Equipped with this definition, we can define the notion of a $t$-type of a color class $C$,
which is simply the subset of profiles at $t$, such that $G_t$ contains a potentially bad clique with that $C$-profile.
It immediately follows that the number of $t$-types is $2^{2^{\calO(w)}}$.
Now, colorings $\calC_t$ of $G_t$ are described by their \emph{$t$-signature}, 
which records how many color classes of each type $\calC_t$ has.
There are at most $k^{f(w)}$ many $t$-signatures, where $f(w) = 2^{2^{\calO(w)}}$,
and this essentially bounds the runtime of the resulting algorithm to 
$n\cdot k^{f(w)} = n^{\calO(f(w))}$.

At the root node $\rootnode \in V(T)$, 
there is only one equivalence class, namely $V_{\rootnode} = V(G)$,
and if in a coloring, there is a clique that is potentially bad for some color class,
then it is indeed a monochromatic maximal clique.
Therefore, at the root node, we only have to check whether there is a coloring all of whose color classes have no
potentially bad cliques.

\newcommand\calS{\ensuremath\mathcal{S}}
\newcommand\insiders{\ensuremath\mathcal{Q}}
\newcommand\outsiders{\ensuremath\mathcal{P}}
\newcommand\bubblebuddies{\ensuremath\mathsf{bb}}
\newcommand\prmergeaux{H'}
\subsection{Potentially Bad Cliques}
We now introduce the main concept used to describe color classes in partial solutions of our algorithms, 
namely potentially bad cliques. 
These are cliques that are monochromatic in some subgraph induced by a set of equivalence classes.
\begin{definition}[Potentially Bad Clique]
	Let $G$ be a graph with rooted branch decomposition $(T, \decf)$ and let $t \in V(T)$.
	A clique $X$ in $G_t$ is called \emph{\eqcmax (in $G_t$)} if it is maximal in $G_t[V(\eqc_t(X))]$.
	Let $C \subseteq V_t$ and let $X$ be a clique in $G_t$.
	Then, $X$ is called \emph{potentially bad for $C$ (in $G_t$)}, if $X$ is \eqcmax in $G_t$ and $X \subseteq C$.
\end{definition}

Naturally, it is not feasible to keep track of \emph{all} potentially bad cliques.
We therefore capture the most vital information about potentially bad cliques in the following 
notion of a \emph{profile}.
For our algorithm, it is only important to know for a color class whether or not it has 
some potentially bad clique with a given profile, rather than how many, or what its vertices are.
This is key to reduce the amount of information we need to store about partial solutions.

There are two components of a profile of a potentially bad clique $X$;
the first one is the set of equivalence classes $\insiders$ containing its vertices,
and the second one consists of the equivalence classes $P \notin \insiders$ 
that have a vertex that is complete to $X$.
This is because, at a later stage, $P$ may be merged with an equivalence class containing vertices of $X$ 
(via the bubbles), in which case $X$ is no longer potentially bad.
We illustrate the following definition in Figure~\ref{fig:ccol:profile}.
\begin{figure}
	\centering
	\includegraphics[height=.1\textheight]{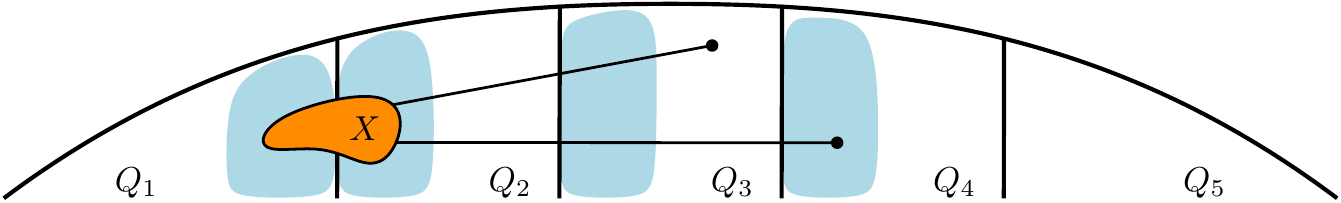}
	\caption{Illustration of the $C$-profile of a clique $X$ that is potentially bad for a color class $C$,
		depicted as the shaded areas within the equivalence classes.
		In this case, we have that $\profile(X \mid C) = (\{Q_1, Q_2\}, \{Q_3, Q_4\})$.}
	\label{fig:ccol:profile}
\end{figure}
\begin{definition}[Profile]
	Let $G$ be a graph with rooted branch decomposition $(T, \decf)$ and let $t \in V(T)$.
	Let $C \subseteq V_t$ and let $X$ be a clique in $G_t$ that is potentially bad for $C$.
	The \emph{$C$-profile} of $X$ is a pair of subsets of $V_t/{\sim_t}$,
	$\profile(X \mid C) \defeq (\insiders, \outsiders)$, where
	\begin{align*}
		\insiders = \eqc_t(X) \mbox{ and }
		\outsiders = \{P \in \eqcbar_t(X) \mid \exists v \in P \colon X \subseteq N(v)\}.
	\end{align*}
	We call the set of all pairs of disjoint subsets of $V_t/{\sim_t}$, 
	where the first coordinate is nonempty, 
	the \emph{profiles at $t$}, formally, 
	\(\profiles_t \defeq \{(\insiders, \outsiders) \mid
			\insiders, \outsiders \subseteq V_t/\sim_t \colon 
			\insiders \neq \emptyset \wedge
			\insiders \cap \outsiders = \emptyset\}.\)
\end{definition}

\begin{observation}\label{obs:ccol:number:of:profiles}
	Let $(T, \decf)$ be a rooted branch decomposition. For each $t \in V(T)$,
	there are at most $2^{\calO(\givenmw)}$ profiles at $t$, where $\givenmw = \modulew(T, \decf)$.
\end{observation}

Let $t \in V(T) \setminus \leaves(T)$ be an internal node with children $r$ and $s$ 
and operator $(\decaux_t, \bubblemap_r, \bubblemap_s)$,
and let $\profile_r \in \profiles_r$ and $\profile_s \in \profiles_s$ be a pair of profiles.
We are now working towards a notion that precisely captures 
when and how a potentially bad clique in $G_r$ for some $C_r \subseteq V_r$ with $C_r$-profile $\pr_r$ 
can be merged with a potentially bad clique in $G_s$ for some $C_s \subseteq V_s$ with $C_s$-profile $\pr_s$
to obtain a potentially bad clique for $C_r \cup C_s$ in $G_t$.
As it turns out, if this is possible, then the profile of the resulting clique 
only depends on $\pr_r$, $\pr_s$, and the operator of $t$.
Note that for now, we focus on the case when the cliques in $G_r$ and $G_s$ are both nonempty,
and we discuss the case when one of them is empty below. 

Before we proceed with this description, we need to introduce some more concepts.
We illustrate all of the following concepts in Figure~\ref{fig:bb:comp:merge-profile}.
\begin{figure}
	\centering
	\includegraphics[height=.25\textheight]{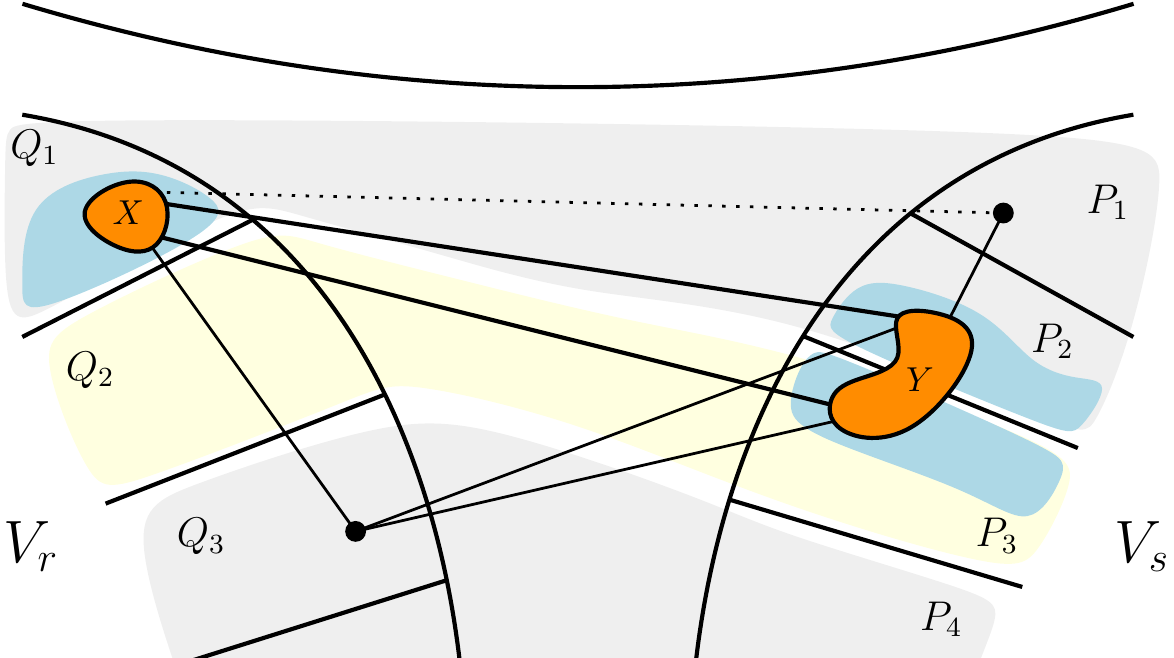}
	\caption{Merging a potentially bad clique $X$ in $G_r$ with a potentially bad clique $Y$ in $G_s$ to obtain 
		a potentially bad clique in $G_t$.
		The color class at hand is depicted in blue and the gray and yellow areas show the (three) bubbles.
		Note that the equivalence classes $P_1$ and $Q_2$ are bubble buddies of $\eqc_r(X)$ and $\eqc_s(Y)$.
		Moreover, the types of $X$ and $Y$ are compatible, 
		since $\{Q_1, P_2, P_3\}$ is a maximal biclique in $\decaux_t[\{Q_1, P_1, P_2, P_3\}]$.
		Finally, note that the equivalence class of $\sim_t$ corresponding to the bubble containing $Q_3$
		will have a vertex that is complete to the potentially bad clique $X \cup Y$.
		}
	\label{fig:bb:comp:merge-profile}
\end{figure}
For a set of equivalence classes $\calS \subseteq V_r/{\sim_r} \cup V_s/{\sim_s}$, its \emph{bubble buddies at $t$}, 
denoted by $\bubblebuddies_t(\calS)$, are the equivalence classes of $V_r/{\sim_r} \cup V_s/{\sim_s}$ 
that are in the same bubble as some equivalence class in~$\calS$:
$$\bubblebuddies_t(\calS) \defeq \{Q_p \mid p \in \{r, s\}, \bubblemap_p(Q_p) \in \bubblemap_r(\calS \cap V_r/{\sim_r}) \cup \bubblemap_s(\calS \cap V_s/{\sim_s})\}.$$
We say that $\pr_r = (\insiders_r, \outsiders_r)$ and $\pr_s = (\insiders_s, \outsiders_s)$ are \emph{compatible}, if
$\insiders_r \cup \insiders_s$ is a maximal biclique in 
\begin{align}
	\label{eq:ccol:pr:compatibility}
	\prmergeaux_t(\pr_r, \pr_s) \defeq \decaux_t[(\insiders_r \cup \insiders_s) \cup ((\outsiders_r \cup \outsiders_s) 
		\cap \bubblebuddies_t(\insiders_r \cup \insiders_s))].
\end{align}
As we show below, the notion of compatibility precisely captures the `merging behavior' of potentially  bad cliques.
Moreover, for $\pr_r$ and $\pr_s$ compatible, we can immediately construct the profile of the resulting potentially bad clique:
the \emph{merge profile} of $\pr_r$ and $\pr_s$ is the profile 
$\mergeprofile(\pr_r, \pr_s) = (\insiders_t, \outsiders_t)$ such that 
\begin{itemize}
	\item $\insiders_t = \bubblemap_r(\insiders_r) \cup \bubblemap_s(\insiders_s)$ and
	\item $\outsiders_t = \bigcup_{\{o, p\} = \{r, s\}} \{\bubblemap(Q_p) \mid 
			Q_p \in \outsiders_p \setminus \bubblebuddies_t(\insiders_r \cup \insiders_s) \colon 
			\insiders_o \subseteq N_{\decaux_t}(Q_p)\}$.
\end{itemize}
\begin{lemma}\label{lem:ccol:merge:profile:bt}
	Let $t \in V(T) \setminus \leaves(T)$ be an internal node with children $r$ and $s$ 
	and operator $(\decaux_t, \bubblemap_r, \bubblemap_s)$.
	For all $p \in \{r, s\}$, let $C_p \subseteq V_p$, 
	let $X_p$ be a clique in $G_r$ that is potentially bad for $C_p$, 
	and let $\pr_p \defeq \profile(X_p \mid C_p) = (\insiders_p, \outsiders_p)$.
	If $\pr_r$ and $\pr_s$ are compatible, 
	then $X_t \defeq X_r \cup X_s$ is a clique that is potentially bad for $C_t \defeq C_r \cup C_s$,
	and $\profile(X_t \mid C_t) = \mergeprofile(\pr_r, \pr_s)$.
\end{lemma}
\begin{proof}
	We first argue that $X_t$ is a clique.
	Since $X_r$ and $X_s$ are cliques, we only have to show that for each $v_r \in X_r$ and $v_s \in X_s$, $v_r v_s \in E(G_t)$.
	In other words, if $Q_r$ is the equivalence class of~$\sim_r$ containing~$v_r$, 
	and $Q_s$ is the equivalence class of~$\sim_s$ containing $v_s$, then $Q_r Q_s \in E(\decaux_t)$.
	Now, $Q_r \in \eqc_r(X_r) = \insiders_r$ and $Q_s \in \eqc_s(X_s) = \insiders_s$,
	and since $\pr_r$ and $\pr_s$ are compatible, we have that $\insiders_r \cup \insiders_s$ is a biclique in $\decaux_t$,
	therefore $Q_r Q_s \in E(\decaux_t)$.
	
	Next, we show that $X_t$ is potentially bad for $C_t$.
	Since $X_r$ and $X_s$ are potentially bad for $C_r$ and $C_s$, respectively, 
	we have that $X_r \subseteq C_r$ and $X_s \subseteq C_t$, and therefore $X_t = X_r \cup X_s \subseteq C_r \cup C_s = C_t$.
	It remains to show that $X_t$ is \eqcmax.
	Suppose not, and let $y \in V(\eqc_t(X_t))$ be a vertex that is complete to $X_t$.
	First, we know that $y \notin V(\eqc_r(X_r) \cup \eqc_s(X_s))$, 
	for if $y \in V(\eqc_p(X_p))$ for some $p \in \{r, s\}$, then $X_p$ is not \eqcmax,
	contradicting $X_p$ being potentially bad for $C_p$.
	On the other hand, we have that 
	$\eqc_t(X_t) = \bubblebuddies_t(\eqc_r(X_r) \cup \eqc_s(X_s)) = \bubblebuddies_t(\insiders_r \cup \insiders_s)$.
	We may assume that for some $p \in \{r, s\}$, 
	the vertex $y$ is contained in some 
	$Q_p \in \bubblebuddies_t(\insiders_r \cup \insiders_s) \setminus (\insiders_r \cup \insiders_s)$.
	Assume up to renaming that $p = r$.
	Since $y$ is complete to $X_t$, we have that $y$ is complete to $X_r$, and therefore $Q_r \in \outsiders_r$.
	In other words, $Q_r$ is contained in the graph $\prmergeaux_t(\pr_r, \pr_s)$ as described in Equation~\eqref{eq:ccol:pr:compatibility}. 
	Moreover, since $y$ is complete to $X_s$, we have that $Q_r$ is complete to $\eqc_s(X_s) = \insiders_s$.
	This implies that $\{Q_r\} \cup \insiders_r \cup \insiders_s$ is a biclique in $\prmergeaux_t(\pr_r, \pr_s)$,
	contradicting $\pr_r$ and $\pr_s$ being compatible.
	
	To conclude the proof, we need to show that $\profile(X_t \mid C_t) = \mergeprofile(\pr_r, \pr_s)$.
	Let $\mergeprofile(\pr_r, \pr_s) = (\insiders_t, \outsiders_t)$.
	We first show that $\eqc_t(X_t) = \insiders_t$.
	To see that $\insiders_t = \bubblemap_r(\insiders_r) \cup \bubblemap_s(\insiders_s) \subseteq \eqc_t(X_t)$,
	we observe that for all $Q_p \in \insiders_p$, there is an $x \in X_p \cap Q_p$.
	This means that $x \in \bubblemap_p(Q_p)$, therefore $X_t \cap \bubblemap_p(Q_p) \neq \emptyset$ and $\bubblemap_p(Q_p) \in \eqc_t(X_t)$.
	The other inclusion can be argued similarly.
	
	Now suppose that $Q_t \in \outsiders_t$.
	Then, $Q_t = \bubblemap_p(Q_p)$ for some $Q_p \in \outsiders_p \setminus \bubblebuddies_t(\insiders_r \cup \insiders_s)$ 
	with $\insiders_o \subseteq N_{\decaux_t}(Q_p)$.
	In other words, there is a vertex $v \in Q_p$ that is complete to $X_t$, 
	and $\bubblemap_p(Q_p) \notin \eqc_t(X_t)$.
	According to the definition of a profile, 
	$Q_t = \bubblemap_p(Q_p)$ is contained in the second coordinate of $\pr_t$.
	The other inclusion can be shown similarly.
\end{proof}

Now we show the other direction, i.e.\ that if we have a potentially bad clique for some $C_t \subseteq V_t$ in $G_t$,
then its restrictions to $V_r$ and $V_s$ necessarily also form potentially bad cliques for the restriction of $C_t$ to
$V_r$ and $V_s$ in $G_r$ and $G_s$, respectively.
Furthermore, in that case, the profiles of the resulting cliques are compatible.
\begin{lemma}\label{lem:ccol:merge:profile:tb}
		Let $t \in V(T) \setminus \leaves(T)$ be an internal node with children $r$ and $s$ 
		and operator $(\decaux_t, \bubblemap_r, \bubblemap_s)$.
		Let $C_t \subseteq V_t$, and let $X_t$ be a clique in $G_t$ that is potentially bad for $C_t$.
		For all $p \in \{r, s\}$, let $X_p \defeq X_t \cap V_p$ and $C_p \defeq C_t \cap V_p$.
		Suppose that for all $p \in \{r, s\}$, $X_p \neq \emptyset$.
		Then, for all $p \in \{r, s\}$, $X_p$ is a potentially bad clique for $C_p$,
		and $\profile_r \defeq \profile(X_r \mid C_r)$ and $\profile_s \defeq \profile(X_s \mid C_s)$ are compatible.
\end{lemma}
\begin{proof}
	Since $X_t$ is a potentially bad clique for $C_t$, we have that $X_t \subseteq C_t$, 
	and so for $p \in \{r, s\}$, $X_p \subseteq C_p$.
	It remains to show that $X_p$ is \eqcmax for all $p \in \{r, s\}$.
	Up to renaming, it suffices to show that $X_r$ is \eqcmax.
	Suppose not and let $y \in \eqc_r(X_r)$ be a vertex that is complete to $X_r$.
	Since $X_t$ is a clique in $G_t$, we have that $\eqc_r(X_r) \cup \eqc_s(X_s)$ is a biclique in $\decaux_t$.
	Therefore, $y$ is also complete to $X_s$ and therefore to $X_t$.
	Clearly, $y \in \eqc_t(X_t)$, and we have a contradiction with $X_t$ being \eqcmax.
	
	What remains to be shown is that $\pr_r = (\insiders_r, \outsiders_r)$ and $\pr_s = (\insiders_s, \outsiders_s)$
	are compatible.
	We have already argued that $\insiders_r \cup \insiders_s = \eqc_r(X_r) \cup \eqc_s(X_s)$ is a biclique in $\decaux_t$;
	we have to show that $\insiders_r \cup \insiders_s$ is a maximal biclique in $\prmergeaux_t \defeq \prmergeaux_t(\pr_r, \pr_s)$ 
	as defined in Equation~\eqref{eq:ccol:pr:compatibility}.
	Clearly, $\insiders_r \cup \insiders_s \subseteq V(\prmergeaux_t)$,
	so suppose that $\insiders_r \cup \insiders_s$ is not a maximal biclique in $\prmergeaux_t$.
	This means that for some $p \in \{r, s\}$, 
	there is some $Q_p \in \outsiders_p \cap \bubblebuddies_t(\insiders_r \cup \insiders_s)$ 
	such that $\{Q_p\} \cup \insiders_r \cup \insiders_s$ is a biclique in $\prmergeaux_t$.
	In that case, there is a vertex $y \in Q_p$ that is complete to $X_t$ 
	(since $Q_p \in \outsiders_p$ and $\{Q_p\} \cup \insiders_r \cup \insiders_s$ is a biclique), 
	and $y \in V(\eqc_t(X_t))$ (since $Q_p \in \bubblebuddies_t(\insiders_r \cup \insiders_s)$);
	we obtained a contradiction with $X_t$ being \eqcmax.
\end{proof}

As mentioned above, we treat the case when a clique $X_p$ in one of the children $p \in \{r, s\}$
remains potentially bad in $G_t$ separately. This is because in that case,
the notion of a maximal biclique in $\prmergeaux_t$ as defined in Equation~\eqref{eq:ccol:pr:compatibility}
does not hold up very naturally.
We formulate the analogous requirements for this case here,
and we skip some of the details.

Let $t \in V(T) \setminus \leaves(T)$ be an internal node with children $r$ and $s$
and operator $(\decaux_t, \bubblemap_r, \bubblemap_s)$.
Let $\pr_r \in \profiles_r$. We say that $\pr_r$ is \emph{liftable} if 
\begin{itemize}
	\item there is no $Q_s \in \bubblebuddies_t(\insiders_r)$ 
		that is complete to $\insiders_r$ in $\decaux_t$, and
	\item $\bubblebuddies_t(\insiders_r) \cap \outsiders_r = \emptyset$.
\end{itemize}
The \emph{lift profile} of $\pr_r$, denoted by $\liftprofile(\pr_r)$, 
is constructed as the merge profile of $\pr_r$ with the empty set;
i.e.\ we take $(\insiders_s, \outsiders_s) = (\emptyset, V_s/{\sim_s})$ and apply the definition given above.
\begin{lemma}\label{lem:ccol:lift:profile}
	Let $t \in V(T) \setminus \leaves(T)$ be an internal node with children $r$ and $s$.
	Let $C_r \subseteq V_r$, $C_s \subseteq V_s$, 
	let $X_r$ be a clique in $G_r$, and let $\pr_r \defeq \profile(X_r \mid C_r)$.
	Then, $X_r$ is a potentially bad clique for $C_r \cup C_s$ in $G_t$ if and only if 
	$X_r$ is a potentially bad clique for $C_r$ in $G_r$ and $\pr_r$ is liftable,
	in which case $\profile_t(X_r \mid C_r \cup C_s) = \liftprofile(\pr_r)$.
\end{lemma}
\begin{proof}
	The proof can be done with very similar arguments to those given above and is therefore omitted.
	One only needs to observe that the notion of `liftable' modulates the notion of a 
	profile being compatible with the profile of an empty set.
\end{proof}

\subsection{The type of a color class}
We now describe the $t$-type of a color class $C$, which is the subset of profiles at $t$
such that there is a clique in $G_t$ that is potentially bad for $C$, with that $C$-profile.
For our algorithm, two color classes with the same type will be interchangeable,
therefore we only have to remember the number of color classes of each type.
\begin{definition}[$t$-Type]
	Let $G$ be a graph with rooted branch decomposition $(T, \decf)$, and let $t \in V(T)$.
	For a set $C \subseteq V_t$, the \emph{$t$-type} of $C$, denoted by $\ctype_t(C)$ is
	\begin{align*}
		\ctype_t(C) \defeq \{\pr_t \in \profiles_t \mid \exists \mbox { clique $X$ in $G_t$ which is potentially bad for $C$ and } \profile(X \mid C) = \pr_t\}.
	\end{align*}
	We call the set $\ctypes_t = 2^{\profiles_t}$ of all subsets of profiles at $t$ the \emph{$t$-types}.
\end{definition}

Since for each $t \in V(T)$, $\card{\profiles_t} \le 2^{\calO(w)}$ by Observation~\ref{obs:ccol:number:of:profiles}, 
the number of $t$-types can be upper bounded as follows.
\begin{observation}\label{obs:ccol:number:of:types}
	Let $(T, \decf)$ be a rooted branch decomposition, and let $t \in V(T)$.
	There are at most $2^{2^{\calO(w)}}$ many $t$-types, where $w \defeq \modulew(T, \decf)$.
\end{observation}

In our algorithm we want to be able to determine the $t$-type of the union of a color class
in $G_r$ and a color class in $G_s$.
This is done via the following notion of a merge type,
which is based on the notion of merge and lift profiles given in the previous section.
\begin{definition}[Merge Type]
	Let $G$ be a graph with rooted branch decomposition $(T, \decf)$, 
	let $t \in V(T) \setminus \leaves(T)$ with children $r$ and $s$.
	For a pair of an $r$-type $\ctype_r \in \ctypes_r$ and an $s$-type $\ctype_s \in \ctypes_s$,
	the \emph{merge type} of $\ctype_r$ and $\ctype_s$, denoted by $\mergetype(\ctype_r, \ctype_s)$, 
	is the $t$-type obtained as follows.
	\begin{align*}
		\mergetype(\ctype_r, \ctype_s) \defeq\;&\left\lbrace\mergeprofile(\pr_r, \pr_s) \mid \pr_r \in \ctype_r,~\pr_s \in \ctype_s, 
			\mbox{where } \pr_r \mbox{ and } \pr_s \mbox{ are compatible}\right\rbrace \\
		  &\bigcup\nolimits_{p \in \{r, s\}} \left\lbrace\liftprofile(\pr_p) \mid \pr_p \in \ctype_p, 
			\mbox{where $\pr_p$ is liftable}\right\rbrace
	\end{align*}
\end{definition}

\begin{lemma}\label{lem:ccol:merge:type}
	Let $G$ be a graph with rooted branch decomposition $(T, \decf)$, 
	let $t \in V(T) \setminus \leaves(T)$ with children $r$ and $s$.
	Let $C_r \subseteq V_r$ and $C_s \subseteq V_s$.
	Then, $\ctype_t(C_r \cup C_s) = \mergetype(\ctype_r(C_r), \ctype_s(C_s))$.
\end{lemma}
\begin{proof}
	Let $C_t \defeq C_r \cup C_s$.
	For one inclusion, let $\pr_t \in \ctype_t(C_t)$. 
	Then, there is a clique $X_t$ in $G_t$ that is potentially bad for $C_t$ whose $C_t$-profile is $\pr_t$.
	If for all $p \in \{r, s\}$, $X_p \defeq X_t \cap V_p \neq \emptyset$, 
	then by Lemma~\ref{lem:ccol:merge:profile:tb},
	we know that for all $p \in \{r, s\}$, $X_p$ is a potentially bad clique for $C_p \defeq C_t \cap V_p$,
	therefore $\pr_p \defeq \profile(X_p \mid C_p) \in \ctype_p(C_p)$.
	Moreover, the lemma asserts that $\pr_r$ and $\pr_s$ are compatible,
	so by construction, we can conclude that $\pr_t = \mergeprofile(\pr_r, \pr_s) \in \mergetype(\ctype_r(C_r), \ctype_s(C_s))$.
	On the other hand, if for some $p \in \{r, s\}$, $X_t \subseteq V_p$, 
	then by Lemma~\ref{lem:ccol:lift:profile}, $X_t$ is a potentially bad clique for $C_p$,
	so $\pr_p \defeq \profile_p(X_t \mid C_p) \in \ctype(C_p)$.
	The lemma also asserts that $\pr_p$ is liftable and that $\liftprofile(\pr_r) = \pr_t$, 
	in which case we also have that $\pr_t \in \mergetype(\ctype_r(C_r), \ctype_s(C_s))$.
	We have argued that $\ctype_t(C_t) \subseteq \mergetype(\ctype_r(C_r), \ctype_s(C_s))$.
	
	For the other inclusion, suppose that $\pr_t \in \mergetype(\ctype_r(C_r), \ctype_s(C_s))$.
	Then, either there is a pair of profiles $\pr_r \in \ctype_r(C_r)$, $\pr_s \in \ctype_s(C_s)$
	such that $\pr_r$ and $\pr_s$ are compatible and $\pr_t = \mergeprofile(\pr_r, \pr_s)$
	or for some $p \in \{r, s\}$, there is a profile $\pr_p \in \ctype_p(C_p)$ that is liftable
	and $\pr_t = \liftprofile(\pr_p)$.
	In the former case, we can use Lemma~\ref{lem:ccol:merge:profile:bt}
	to conclude that $\pr_t \in \ctype_t(C_t)$, 
	and in the latter case, we have that $\pr_t \in \ctype_t(C_t)$ by Lemma~\ref{lem:ccol:lift:profile}.
	This shows that $\mergetype(\ctype_r(C_r), \ctype_s(C_s)) \subseteq \ctype_t(C_t)$ which concludes the proof.
\end{proof}

\subsection{The algorithm}
We are now ready to describe the algorithm. 
As alluded to above, partial solutions at a node $t$, i.e.\ colorings of $G_t$, 
are described via the notion of a \emph{$t$-signature} which records the number of color classes of each type
in a coloring.
If two colorings have the same $t$-signature, then they are interchangeable as far as our algorithm is concerned.
We show that this information suffices to solve the problem in a bottom-up dynamic programming fashion.
\begin{definition}[$t$-Signature]
	Let $k$ be a positive integer.
	Let $G$ be a graph with rooted branch decomposition $(T, \decf)$, 
	let $t \in V(T)$,
	and let $\calC = (C_1, \ldots, C_k)$ be a $k$-coloring of $G_t$.
	Then, $\csig_{\calC} \colon \ctypes_t \to \{0, 1, \ldots, k\}$ where
	\[
		\forall \ctype_t \in \ctypes_t\colon \csig_{\calC}(\ctype_t) \defeq \card{\{i \in \{1, \ldots, k\} \mid \ctype(C_i) = \ctype_t\}}, 
	\]
	is called the \emph{$t$-signature} of $\calC$.
	The set of \emph{$t$-signatures} is defined as:
	\[
		\csignatures_t \defeq \left\{\csig_t \colon \ctypes_t \to \{0, 1, \ldots, k\} \;\middle|\; 
		\sum\nolimits_{\ctype_t \in \ctypes_t} \csig_t(\ctype_t) = k\right\}
	\]
\end{definition}

The following bound on the number of $t$-signatures immediately follows from Observation~\ref{obs:ccol:number:of:types},
stating that the number of $t$-types is upper bounded by $2^{2^{\calO(w)}}$.
\begin{observation}\label{obs:number:of:signatures}
	Let $(T, \decf)$ be a rooted branch decomposition of an $n$-vertex graph, and let $t \in V(T)$.
	There are at most $k^{2^{2^{\calO(w)}}} \le n^{2^{2^{\calO(w)}}}$ many $t$-signatures, 
	where $w \defeq \modulew(T, \decf)$ and $k$ is the number of colors.
\end{observation}

\begin{dptabledef}
	For each $t \in V(T)$ and $\csig_t \in \csignatures_t$, 
	we let $\dptable[t, \csig_t] = 1$ if and only if 
	there is a $k$-coloring $\calC$ of $G_t$ such that $\csig_{\calC} = \csig_t$.
\end{dptabledef}

We now show that the information stored at the table entries suffices to determine whether or not
our input is a \yes-instance; 
that is, after filling all the table entries, we can read off the solution to the problem at the root node.

\begin{lemma}\label{lem:ccol:root}
	Let $G$ be a graph with rooted branch decomposition $(T, \decf)$, and let $\rootnode$ be the root of $T$.
	$G$ has a clique coloring with $k$ colors if and only if $\dptable[\rootnode, \csig^\star] = 1$,
	where $\csig^\star$ is the $\rootnode$-signature for which $\csig^\star(\emptyset) = k$.
\end{lemma}
\begin{proof}
	The lemma immediately follows from two facts. 
	First, since $\csig^\star(\emptyset) = k$, we have that $\csig^\star(\ctype_\rootnode) = 0$ 
	for any other $\rootnode$-type $\ctype_\rootnode \neq \emptyset$.
	Second, that for each set $C \subseteq V_\rootnode = V(G)$,
	the set of potentially bad cliques for $C$ is precisely the set of maximal cliques that are fully contained in $C$,
	i.e.\ it is the set of monochromatic maximal cliques in the corresponding coloring that are contained in $C$.
\end{proof}

We first describe how to compute the table entries at the leaves, by brute-force.
\begin{dpleaves}
	Let $t \in \leaves(T)$ be a leaf node in $T$ and let $v \in V(G)$ 
	be the vertex such that $\decf(v) = t$.
	We show how to compute the table entries $\dptable[t, \cdot]$.
	Note that $G_t = (\{v\}, \emptyset)$, and that $\{v\}$ is the only equivalence class of $\sim_t$.
	To describe the types of color classes of $G_t$, observe that
	the only \eqcmax clique in $G_t$ is $\{v\} \eqdef X_v$, 
	which is potentially bad for $C_v \defeq \{v\} = X_v$.
	In that case, we have that $\pr_v \defeq \profile(X_v \mid C_v) = (\{v\}, \emptyset)$,
	and the type of color class $C_v$ is $\{\pr_v\}$.
	The type of the remaining $k-1$ color classes is $\emptyset$, since they are all empty.
	Therefore, for each $t$-signature $\csig_t$, we set $\dptable[t, \csig_t] \defeq 1$
	if and only if $\csig_t(\{\pr_v\}) = 1$ and  $\csig_t(\emptyset) = k - 1$.
\end{dpleaves}

Next, we move on to the computation of the table entries at internal nodes of the branch decomposition.
To describe this part of the algorithm,
we borrow the following notion of a \emph{merge skeleton} from~\cite{JLL20}.\footnote{Note 
that in~\cite{JLL20}, the graph structure of the bipartite graph plays a role, in that there is only edges between compatible types.
In the present setting, there is no notion of compatibility of color class types which is why the bipartite graph
of the merge skeleton is always complete.}
\begin{definition}[Merge skeleton]
	Let $G$ be a graph and $(T, \decf)$ one of its rooted branch decompositions.
	Let $t \in V(T) \setminus \leaves(T)$ with children $r$ and $s$.
	The \emph{merge skeleton} of $r$ and $s$ is an edge-labeled complete bipartite graph
	$(\mergeaux, \malab)$ where
	\begin{itemize}
		\item $V(\mergeaux) = \ctypes_r \cup \ctypes_s$, and
		\item for all $\ctype_r \in \ctypes_r$, $\ctype_s \in \ctypes_s$, 
			$\malab(\ctype_r \ctype_s) = \mergetype(\ctype_r, \ctype_s)$.
	\end{itemize}
\end{definition}
\begin{algorithm}
	\lForEach{$\csig_t \in \csignatures_t$}{set $\dptable[t, \csig_t] \gets 0$}
	Let $(\mergeaux, \malab)$ be the merge skeleton of $t$\;
	\ForEach{$\csig_r \in \csignatures_r$, $\csig_s \in \csignatures_s$ such that $\dptable[r, \csig_r] = 1$ and $\dptable[s, \csig_s] = 1$}{
		\ForEach{\label{alg:cliquecol:internal:l4}$\maassign \colon E(\mergeaux) \to \{0, 1, \ldots, k\}$ such that 
			\begin{enumerate}
				\item $\sum_{e \in E(\mergeaux)} \maassign(e) = k$, and
				\item for all $p \in \{r, s\}$ and all $\ctype_p \in \ctypes_p$, it holds that
					$\sum_{\ctype_p \ctype_o \in E(\mergeaux)} \maassign(\ctype_p \ctype_o) = \csig_p(\ctype_p)$
			\end{enumerate}
		}{
			Let $\csig_t \colon \ctypes_t \to \{0, 1, \ldots, k\}$ 
			be such that for all $\ctype_t \in \ctypes_t$,
			$\csig_t(\ctype_t) = \sum_{e \in E(\mergeaux), \malab(e) = \ctype_t} \maassign(e)$\label{alg:cliquecol:internal:l6}\;
			update $\dptable[t, \csig_t] \gets 1$\;
		}
	}
	\caption{Algorithm to set the table entries at an internal node $t \in V(T) \setminus \leaves(T)$ with children
		$r$ and $s$, assuming the table entries at $r$ and $s$ have been computed.}
	\label{alg:cliquecol:internal}
\end{algorithm}
\begin{dpinternal}
	Let $t \in V(T) \setminus \leaves(T)$ be an internal node with children $r$ and $s$.
	We discuss how to compute the table entries at $t$, assuming the table entries at $r$ and $s$ have been computed.
	Each coloring of $G_t$ can be obtained from a coloring of $G_r$ and a coloring of $G_s$,
	by merging pairs of color classes.
	Therefore, for each pair $\csig_r \in \csignatures_r$, $\csig_s \in \csignatures_s$ such that 
	$\dptable[r, \csig_r] = 1$ and $\dptable[s, \csig_s] = 1$, we do the following.
	We enumerate all labelings of the edge set of the merge skeleton with numbers from $\{0, 1, \ldots, k\}$,
	with the following interpretation.
	If an edge $\ctype_r \ctype_s$ has label $j$, then it means that $j$ color classes of $r$-type $\ctype_r$
	will be merged with $j$ color classes of $s$-type $\ctype_s$; 
	this gives $j$ color classes of $t$-type $\mergetype(\ctype_r, \ctype_s) = \malab(\ctype_r \ctype_s)$.
	Each such labeling that respects the number of color classes available of each type will produce a 
	coloring of $G_t$ with some signature $\csig_t$, which can then be read off the edge labeling.
	For all such $\csig_t$, we set $\dptable[t, \csig_t] = 1$.
	We give the formal details in Algorithm~\ref{alg:cliquecol:internal}.
\end{dpinternal}

We now prove correctness of the algorithm.
\begin{lemma}\label{lem:ccol:cor}
	Let $G$ be a graph and $(T, \decf)$ one of its rooted branch decompositions,
	and let $t \in V(T)$.
	The above algorithm computes the table entries $\dptable[t, \cdot]$ correctly,
	i.e.\ for each $\csig_t \in \csignatures_t$, it sets $\dptable[t, \csig_t] = 1$
	if and only if $G_t$ has a $k$-coloring $\calC$ with $\csig_\calC = \csig_t$.
\end{lemma}
\begin{proof}
	The proof is by induction on the height of $t$.
	In the base case, when $t$ is a leaf, 
	it is straightforward to verify correctness.
	
	Now suppose that $t \in V(T) \setminus \leaves(T)$ is an internal node with children $r$ and $s$, 
	and let $(\mergeaux, \malab)$ be the merge skeleton at $t$.
	Suppose for some $t$-signature $\csig_t \in \csignatures_t$, the algorithm set $\dptable[t, \csig_t] = 1$.
	Then, there is some $r$-signature $\csig_r$ and some $s$-signature $\csig_s$ 
	such that $\dptable[r, \csig_r] = 1$, $\dptable[s, \csig_s] = 1$, and 
	there is a map $\maassign\colon E(\mergeaux) \to \{0, 1, \ldots, k\}$
	satisfying the conditions of lines~\ref{alg:cliquecol:internal:l4} and~\ref{alg:cliquecol:internal:l6} 
	in Algorithm~\ref{alg:cliquecol:internal}.
	By induction, there is a $k$-coloring $\calC_r$ of $G_r$ whose $r$-signature is $\csig_r$,
	and a $k$-coloring of $\calC_s$ of $G_s$ whose $s$-signature is $\csig_s$.
	We construct the desired coloring $\calC_t$ of $G_t$ whose $t$-signature is $\csig_t$ as follows:
	For each pair of an $r$-type $\ctype_r$ and an $s$-type $\ctype_s$, 
	we take $\maassign(\ctype_r \ctype_s)$ pairs of a color class $C_r$ of $r$-type $\csig_r$
	and a color class $C_s$ of $s$-type $C_s$, and for each such pair, we add $C_r \cup C_s$ 
	as a color class to $\calC_t$.
	By Lemma~\ref{lem:ccol:merge:type}, the $t$-type of $C_r \cup C_s$ is 
	$\mergetype(\ctype_r, \ctype_s) = \malab(\ctype_r \ctype_s)$.
	The condition in line~\ref{alg:cliquecol:internal:l4} ensures that each color class of 
	$\calC_r$ and each color class of $\calC_s$ is used precisely once to create a color class of $\calC_t$,
	and the condition in line~\ref{alg:cliquecol:internal:l6} ensures that 
	the $t$-signature of $\calC_t$ is indeed $\csig_t$.
	
	For the other direction, suppose that there is a $k$-coloring $\calC_t$ of $G_t$ with $t$-signature $\csig_t$.
	We construct a pair of a coloring $\calC_r$ of $G_r$ and a coloring of $\calC_s$ of $G_s$,
	together with their signatures $\csig_r$ and $\csig_s$, respectively,
	and a map $\maassign\colon E(\mergeaux) \to \{0, 1, \ldots, k\}$.
	Initially, for all $p \in \{r, s\}$, we let $\calC_p = \emptyset$, 
	and for all $\ctype_p \in \ctypes_p$, $\csig_p(\ctype_p) \defeq 0$.
	Moreover, we let $\maassign(e) \defeq 0$ for all $e \in E(\mergeaux)$.
	
	For each color class $C_t \in \calC_t$, 
	we add $C_r \defeq C_t \cap V_r$ to $\calC_r$ and $C_s \defeq C_t \cap V_s$ to $\calC_s$.
	Let $\ctype_t$ be the $t$-type of $C_t$.
	By Lemma~\ref{lem:ccol:merge:type}, $C_r$ has some $r$-type $\ctype_r$ and $C_s$ has some $s$-type $\ctype_s$
	such that $\ctype_t$ is the merge type $\mergetype(\ctype_r, \ctype_s)$ of $\ctype_r$ and $\ctype_s$.
	We increase the values of $\csig_r(\ctype_r)$ and $\csig_s(\ctype_s)$ by $1$, 
	since we added one more color class of $r$-type $\ctype_r$ to $\calC_r$, and 
	one more color class of $s$-type $\ctype_s$ to $\calC_s$.
	Additionally, we add $1$ to the value of $\maassign(\ctype_r \ctype_s)$,
	since $C_t$ is a color class of $t$-type $\mergetype(\ctype_r, \ctype_s) = \malab(\ctype_r \ctype_s)$
	obtained from merging $C_r$ (a color class of $r$-type $\ctype_r$)
	with $C_s$ (a color class of $s$-type $\ctype_s$).
	
	After doing this for all color classes of $\calC_t$, 
	we have that $\calC_r$ is a $k$-coloring with $r$-signature $\csig_r$, 
	and that $\calC_s$ is a $k$-coloring with $s$-signature $\csig_s$.
	By induction, $\dptable[r, \csig_r] = 1$ and $\dptable[s, \csig_s] = 1$.
	It remains to argue that $\maassign$ satisfies the conditions 
	expressed in lines~\ref{alg:cliquecol:internal:l4} and~\ref{alg:cliquecol:internal:l6} 
	in Algorithm~\ref{alg:cliquecol:internal}.
	The first item of line~\ref{alg:cliquecol:internal:l4} is clearly satisfied,
	since we increased $\card{\calC_t} = k$ values of $\maassign$ by $1$ in the above process.
	The second item holds since we increased the value of some $\csig_p(\ctype_p)$ by $1$
	if and only if we increased the value of an edge $e$ incident with $\ctype_p$ in $\mergeaux$ by $1$.
	To see that for each $\ctype_t$, 
	$\csig_t(\ctype_t) = \sum_{e \in E(\mergeaux), \malab(e) = \ctype_t} \maassign(e)$,
	observe that we identified for each color class of type $\ctype_t$,
	the occurrence of $\ctype_t$ as a merge type of a pair of an $r$-type and an $s$-type,
	and therefore a label of some edge $e \in E(\mergeaux)$, 
	and increased $\maassign(e)$ by $1$ in such a case. 
	We can conclude that $\csig_t$ can be obtained as shown in line~\ref{alg:cliquecol:internal:l6} 
	of Algorithm~\ref{alg:cliquecol:internal}, and so the algorithm set
	$\dptable[t, \csig_t] = 1$.
\end{proof}

\begingroup
\newcommand\eff{2^{2^{\calO(w)}}}
To wrap up, it remains to argue the runtime of the algorithm.
Suppose we are given a graph $G$ with rooted branch decomposition $(T, \decf)$ 
and let $w \defeq \modulew(T, \decf)$.
By Observation~\ref{obs:number:of:signatures}, there are at most $n^{\eff}$ table entries at each node of $T$.
The entries of leaf nodes can clearly be computed in time $n^{\calO(1)}$.
Now let $t \in V(T) \setminus \leaves(T)$ be an internal node with children $r$ and $s$.
To compute all table entries at $t$, we execute Algorithm~\ref{alg:cliquecol:internal}. 
In the worst case, it loops over each pair of an $r$-signature and an $s$-signature,
and given such a pair, it enumerates all labelings of the edges of the merge skeleton $\mergeaux$ 
with numbers from $\{0, 1, \ldots, k\}$ (such that all entries sum up to $k$).
We have that $\card{E(\mergeaux)} = \card{\ctypes_r} \cdot \card{\ctypes_s} = \left(\eff\right)^2 = \eff$
(see Observation~\ref{obs:ccol:number:of:types}),
therefore the number of labelings to consider is upper bounded by $k^{\eff} \le n^{\eff}$.
The runtime of Algorithm~\ref{alg:cliquecol:internal} can therefore be upper bounded by
\begin{align*}
	\left(n^{\eff}\right)^2 \cdot n^{\eff} = n^{\eff},
\end{align*}
and since $\card{V(T)} = \calO(n)$, this equals the runtime of the whole procedure.
\endgroup
Correctness is proved in Lemma~\ref{lem:ccol:cor}, and 
Lemma~\ref{lem:ccol:root} asserts that the solution to the problem can be read off the table 
entries at the root, once computed. 
Using standard memoization techniques, we can modify the above algorithm so that it returns a coloring if one exists.
We therefore have the following theorem.
\begin{theorem}\label{thm:ccol:xp:cliquew}
	There is an algorithm that given a graph $G$ 
	together with one of its rooted branch decompositions $(T, \decf)$ and a positive integer $k$, 
	decides whether $G$ has a clique coloring with $k$ colors in time $n^{2^{2^{\calO(w)}}}$,
	where $w \defeq \modulew(T, \decf)$.
	If such a coloring exists, the algorithm can construct it.
\end{theorem}

\bibliographystyle{plain}
\bibliography{ref}

\end{document}